\def\year{2021}\relax
\documentclass[letterpaper]{article} 
\usepackage{aaai21}  
\usepackage{times}  
\usepackage{helvet} 
\usepackage{courier}  
\usepackage[hyphens]{url}  
\usepackage{graphicx} 
\usepackage{natbib}  
\usepackage{caption} 
 \usepackage{hyperref}

\usepackage{amsthm}
\usepackage{amssymb}  
\usepackage{bm}
\usepackage{amsmath}
\usepackage{color}
\usepackage{verbatim}

\usepackage{array} 
\newcolumntype{C}[1]{>{\centering\let\newline\\\arraybackslash\hspace{0pt}}m{#1}}
\usepackage{thmtools,thm-restate} 
\usepackage{ifthen} 
\newcommand{\camera}{0}

\usepackage[autostyle=true,german=quotes]{csquotes}

\usepackage[switch]{lineno}

\newtheorem{theorem}{Theorem}[section]
\newtheorem{lemma}[theorem]{Lemma}
\newtheorem{corollary}[theorem]{Corollary}

\newtheorem{example}[theorem]{Example}

\newcommand{\p}[1]{\left( #1 \right)}
\newcommand{\br}[1]{\left[ #1 \right]}

\newcommand{\cd}[0]{\cdot}

\newcommand{\nplayer}[0]{\ensuremath{M}}
\newcommand{\gendist}[0]{\ensuremath{\Theta}}
\newcommand{\mean}[0]{\ensuremath{\theta}}
\newcommand{\mue}[0]{\ensuremath{\mu_e}}
\newcommand{\var}[0]{\ensuremath{\sigma^2}}
\newcommand{\sampledist}[0]{\ensuremath{\mathcal{D}}}
\newcommand{\xdist}[1]{\ensuremath{\mathcal{X}_{#1}}}
\newcommand{\param}[0]{\ensuremath{\bm{\theta}}}
\newcommand{\dimval}[0]{\ensuremath{D}}
\newcommand{\ndraw}[0]{\ensuremath{n}}
\newcommand{\err}[0]{\ensuremath{\epsilon^2}}
\newcommand{\expdata}[1]{\ensuremath{\mathbb{E}_{\Ymf \sim \sampledist(\mean_{#1}, \err_{#1})}}}
\newcommand{\expparam}[1]{\ensuremath{\mathbb{E}_{(\mean_{#1}, \err_{#1}) \sim \gendist}}}
\newcommand{\total}[0]{\ensuremath{N}}
\newcommand{\expXlin}[1]{\ensuremath{\mathbb{E}_{X_{#1} \sim \xdist{#1}}}}
\newcommand{\expxlin}[1]{\ensuremath{\mathbb{E}_{x \sim \xdist{#1}}}}
\newcommand{\expetalin}[1]{\ensuremath{\mathbb{E}_{\errval_{#1} \sim \sampledist_{#1}(0, \err_{#1})}}}
\newcommand{\expYlin}[1]{\ensuremath{\mathbb{E}_{Y_{#1} \sim \sampledist_{#1}(X_{#1}^T \bm{\param}_{#1}, \err_{#1})}}}
\newcommand{\errval}[0]{\ensuremath{\bm{\eta}}}
\newcommand{\errvalmf}[0]{\ensuremath{\eta}}
\newcommand{\X}[0]{\ensuremath{\bm{X}}}
\newcommand{\x}[0]{\ensuremath{\bm{x}}}
\newcommand{\Y}[0]{\ensuremath{\bm{Y}}}
\newcommand{\Ymf}[0]{\ensuremath{Y}}
\newcommand{\cross}[1]{\ensuremath{\Sigma_{#1}}}
\newcommand{\crossc}[1]{\ensuremath{Cov_{#1}}}
\newcommand{\norm}[1]{\left\lVert#1\right\rVert}
\newcommand{\w}[0]{\ensuremath{w}}
\newcommand{\vm}[0]{\ensuremath{v}}
\newcommand{\alone}[0]{\ensuremath{\pi_l}}
\newcommand{\gcol}[0]{\ensuremath{\pi_g}}
\newcommand{\sepcol}[0]{\ensuremath{\pi_d}}

\newcommand{\s}[0]{\ensuremath{s}}
\newcommand{\el}[0]{\ensuremath{\ell}}
\newcommand{\ns}[0]{\ensuremath{n_{\s}}}
\newcommand{\nl}[0]{\ensuremath{n_{\el}}}
\newcommand{\nlv}[0]{\ensuremath{n_{\el}}}
\newcommand{\Sv}[0]{\ensuremath{S}}
\newcommand{\Lv}[0]{\ensuremath{L}}
\newcommand{\si}[0]{\ensuremath{S}}
\newcommand{\li}[0]{\ensuremath{L}}
\newcommand{\col}[0]{\ensuremath{C}}

\newcommand{\cgeq}[0]{\ensuremath{\succeq}}
\newcommand{\cg}[0]{\ensuremath{\succ}}
\newcommand{\cleq}[0]{\ensuremath{\preceq}}
\newcommand{\cl}[0]{\ensuremath{\prec}}

\title{Model-sharing Games:\\ Analyzing Federated Learning Under Voluntary Participation}
\author{
        Kate Donahue,\textsuperscript{\rm 1}
        Jon Kleinberg, \textsuperscript{\rm 1, 2}\\
}
\affiliations {
    \textsuperscript{\rm 1} Department of Computer Science, Cornell University \\
    \textsuperscript{\rm 2} Department of Information Science, Cornell University \\
    kdonahue@cs.cornell.edu, 
    kleinber@cs.cornell.edu
}
\nocopyright 
\begin{document}

\maketitle
\thispagestyle{plain}
\pagestyle{plain}

\begin{abstract}
Federated learning is a setting where agents, each with access to their own data source, combine models learned from local data to create a global model. If agents are drawing their data from different distributions, though, federated learning might produce a biased global model that is not optimal for each agent. This means that agents face a fundamental question: should they join the global model or stay with their local model? In this work, we show how this situation can be naturally analyzed through the framework of coalitional game theory.

Motivated by these considerations, we propose the following game: there are heterogeneous players with different model parameters governing their data distribution and different amounts of data they have noisily drawn from their own distribution. Each player's goal is to obtain a model with minimal expected mean squared error (MSE) on their own distribution. They have a choice of fitting a model based solely on their own data, or combining their learned parameters with those of some subset of the other players. Combining models reduces the variance component of their error through access to more data, but increases the bias because of the heterogeneity of distributions.

In this work, we derive exact expected MSE values for problems in linear regression and mean estimation. We use these values to analyze the resulting game in the framework of hedonic game theory; we study how players might divide into coalitions, where each set of players within a coalition jointly construct model(s). We analyze three methods of federation, modeling differing degrees of customization. In uniform federation, the agents collectively produce a single model. In coarse-grained federation, each agent can weight the global model together with their local model. In fine-grained federation, each agent can flexibly combine models from each other agent in the federation. For each method, we constructively analyze the stable partitions of players into coalitions. 
\end{abstract}

\section{Introduction}

Imagine a situation as follows: a hospital is trying to evaluate the effectiveness of a certain procedure based on data it has collected from procedures done on patients in their facilities. It seems likely that certain attributes of the patient influences the effectiveness of the procedure, so the hospital analysts opt to fit a linear regression model with parameters $\hat{\param}$. However, because of the limited amount of data the hospital has access to, this model has relatively high error. Luckily, other hospitals also have data from implementations of this same procedure. However, for reasons of privacy, data incompatibility, data size, or other operational considerations, the hospitals don't wish to share raw patient data. Instead, they they opt to combine their models by taking a weighted average of the parameters learned by each hospital. If there are $\nplayer$ hospitals and hospital $i$ has $\ndraw_i$ samples, the combined model parameters would look like: 
$$\hat{\param}^f = \frac{1}{\sum_{i=1}^{\nplayer}\ndraw_i} \sum_{i=1}^{\nplayer}\hat{\param}_i \cd \ndraw_i$$

The situation described above could be viewed as a stylized model of \emph{federated learning}. Federated learning is a distributed learning process that is currently experiencing rapid innovations and widespread implementation \cite{Li_2020, kairouz2019advances}. It is used in cases where data is distributed across multiple agents and cannot be combined centrally for training. For example, federated learning is implemented in word prediction on cell phones, where transferring the raw text data would be infeasible given its large size (and sensitive content). The motivating factor for using federated learning is that access to more data will reduce the variance in a learned model, reducing its error.

However, there could be a downside to using federated learning. In the hospital example, it seems quite reasonable that certain hospitals might have different true generating models for their data, based on the differences in patient populations or variants of the procedure implementation, for example. Two dissimilar hospitals that are federating together will see a decrease in their model's error due to model variance, but an increase in their error due to model bias. This raises some fundamental questions for each participating hospital - or, more generally, each agent $i$ considering federated learning. Which other agents should $i$ federate with in order to minimize its error? Will those other agents be interested in federating with $i$? Does there exist some stable arrangement of agents into federating clusters, and if so, what does that arrangement look like? 

Numerous works have explored the issue of heterogeneous data in federated learning - we discuss specifically how they relate to ours in a later section. Often the goal in these lines of work is to achieve equality in error rates guaranteed to each agent, potentially by actively collecting more data or using transfer learning to ensure the model better fits local data. However, to our knowledge, there has not yet been work that systematically looks at the participation questions inherent in federated learning through the lens of game theory --- especially the theory of {\em hedonic games}, which studies the formation of self-sustaining coalitions.

In a hedonic game, players are grouped together into clusters or coalitions: the overall collection of coalitions is called a \emph{coalition structure}. Each player's utility depends solely on the identity of the other players in its coalition. A common question in hedonic games is the stability of a coalition structure. A coalition structure $\Pi$ is \emph{core-stable} (or \enquote{in the core}) if there does not exist a coalition $\col$ so that every player in $\col$ prefers $\col$ to its coalition in $\Pi$. A coalition structure is \emph{strictly core stable} if there does not exist a coalition $\col$ so that every player in $\col$ weakly prefers $\col$ to its coalition in $\Pi$, and at least one player $\in \col$ strictly prefers $\col$ to $\Pi$. A coalition structure is \emph{individually stable} if there does not exist a coalition $\col \in \Pi$ so that a player $i\not \in \col$ prefers $\col \cup \{i\}$ to its arrangement in $\Pi$ and all players in $\col$ weakly prefer $\col \cup \{i\}$ to $\col$ \cite{BOGOMOLNAIA2002201}. 

To explain the analogy of federated learning to hedonic games, we first consider that each agent in federated learning is a player in a hedonic game. A player is in coalition with other players if it is federating with them. Its cost is its expected error in a given federating cluster, which depends only on the identity of other players in its federating cluster. Players are assumed to be able to move between federating clusters only if doing so would benefit itself and not harm other players in the cluster it is moving to: notably, we allow players to freely leave a cluster, even if doing so would harm the players in the cluster it leaves behind. 

\paragraph{\bf The present work: Analyzing federated learning through hedonic game theory} 

In this work, we use the framework of hedonic games to analyze the stability of coalitions in data-sharing applications that capture key issues in federated learning.  By working through a sequence of deliberately stylized models, we obtain some general insights about participation and stability in these kinds of applications.

For the first case, we analyze \emph{uniform} federation. In this simplest case, a federating cluster produces a single model, which each player uses. For uniform federation, first we consider the case where all players have the same number of data points. We show that in this game, when the number of data points $\ndraw$ is fairly small, the only core-stable coalition structure is to have all players federating together, in the \enquote{grand coalition}. When $\ndraw$ is large, the only core-stable coalition structure is to have all players separate (doing local learning). There exists a point case of intermediate $\ndraw$ size where all coalition structures are core-stable. Next, we analyze the case where all players have either one of two sizes (\enquote{small} or \enquote{large}). The analysis is more complicated, but we demonstrate constructively that there always exists an individually stable partition of players into clusters. 

Besides uniform federation, we also analyze two other forms of federation. For \emph{coarse-grained} federation, the federating cluster still produces a single model, but each player can weight the global model with their local model, allowing some personalization. For coarse-grained federation, when all players have the same number of samples, we show that the grand coalition (all players federating together) is always the only core stable arrangement. For the small/large case, we produce a simple, efficient algorithm to calculate a strictly core-stable arrangement. Additionally, we show that, for this federation method, the grand coalition is always individually stable (no player wishes to unilaterally defect).

Finally, for \emph{fine-grained} federation, each player is allowed to take the local models of other players in the federation and combine them using whichever weights they choose to produce a model customized for their use. With fine-grained federation, we show that the grand coalition is always core stable. 

We are only able to produce these hedonic game theory results because of our derivations of exact error values for the underlying inference problems. We calculate these values for all three methods of federation, and for agents federating in two situation: 1) a mean estimation problem and 2) a linear regression problem. The error values depend on the number of samples each agent has access to, with the expectation taken over the values of samples each agent draws as well as the possible different true parameters of the data each player is trying to model. Our results are completely independent of the generating distributions used, relying only weakly on two parameters.

The results in this paper are theoretical and do not depend on any experiments or code. However, while writing the paper, we found it useful to write and work with code to check conjectures. Some of that code is publicly available at \url{https://github.com/kpdonahue/model_sharing_games}. 

Before moving to the main technical content, the next section will walk through a motivating example, followed by a review of related literature and a description of our model and assumptions. Beyond technical assumptions, recent work (\citet{cooper2020normative}) has highlighted the importance of describing \emph{normative} assumptions researchers make: we also include a summary of the most important normative assumptions of our analysis \ifthenelse{\equal{\camera}{1}}{in our ethics statement after the main text}{}.

\subsection{Related works}
\paragraph{\bf Incentives and federated learning:} \citet{collabPAC} describes an approach to handling heterogeneous data where more samples are iteratively gathered from each agent in a way so that all agents are incentivized to participate in the grand coalition during federated learning. \citet{Duan_selfbalance} builds a framework to schedule data augmentation and resampling. \citet{yu2020salvaging} demonstrates empirically that there can be cases where individuals get lower error with local training than federated and evaluates empirical solutions. \citet{wang2020split} analyzes the question of when it makes sense to split or not to split datasets drawn from different distributions. Finally, \citet{unpublished} analyzes notions of envy and efficiency with respect to sampling allocations in federated learning. 

\paragraph{\bf Transfer learning:} \citet{mansour2020approaches} and \citet{deng2020adaptive} both propose theoretical methods for using transfer learning to minimize error provided to agents with heterogeneous data.  \citet{li2019fair} and \citet{MartinezMinMax} both provide methods to produce a more uniform level of error rates across agents participating in federated learning.

\paragraph{\bf Clustering and federated learning:} 
\citet{Sattler_2020} and \citet{ShlezingerClustFed} provide an algorithm to \enquote{cluster} together players with similar data distributions with the aim of providing them with lower error. They differ from our approach in that they consider the case where there is some knowledge of each player's data distribution, where we only assume knowledge of the number of data points. Additionally, their approach doesn't explicitly consider agents to be game-theoretic actors in the same way that this one does. Interestingly, \citet{Guazzone_2014} uses a game theoretic framework to analyze federated learning, but with the aim of minimizing energy usage, not error rate. 

\section{Motivating example}
\label{motivate}

\begin{table}[]
\centering
\begin{tabular}{|c|c|c|c|}
\hline
Coalition structure & $err_a(\cd)$ & $err_b(\cd)$ & $err_c(\cd)$ \\ \hline
$\{a\},\{b\},\{c\}$ & 2&   2 & 2      \\\hline 
$\{a, b\},\{c\}$ & 1.5& 1.5                & 2             \\ \hline
$\{a, b, c\}$ & 1.3& 1.3 & 1.3              \\ \hline
\end{tabular}
\caption{The expected errors using uniform federation of players in each coalition when all three players have 5 samples each, with parameters $\mue = 10, \var =1$. Each row denotes a different coalition partition: for example, $\{a, b\} \{c\}$ indicates that players $a$ and $b$ are federating together while $c$ is alone.}
\label{tab:allsmall}
\vspace{-4mm}
\end{table}

\begin{table}[]
\centering
\begin{tabular}{|c|c|c|c|}
\hline
Coalition structure & $err_a(\cd)$ & $err_b(\cd)$ & $err_c(\cd)$ \\ \hline
$\{a\},\{b\},\{c\}$ & 2&   2 & 0.4     \\\hline 
$\{a, b\},\{c\}$ & 1.5& 1.5                & 0.4             \\ \hline
$\{a\},\{b, c\}$ & 2& 1.72 & 0.39              \\ \hline
$\{a, b, c\}$ & 1.55& 1.55 & 0.41              \\ \hline
\end{tabular}
\caption{The expected errors using uniform federation of players in each coalition when players $a$ and $b$ have 5 samples each and player $c$ has 25 samples, with parameters $\mue = 10, \var =1$.}
\label{tab:2small1big}
\vspace{-4mm}
\end{table}

\begin{table}[]
\centering
\begin{tabular}{|c|c|c|c|}
\hline
Coalition structure & $err_a(\cd)$ & $err_b(\cd)$ & $err_c(\cd)$ \\ \hline
$\{a\},\{b\},\{c\}$ & 0.4&   0.4 & 0.4     \\\hline 
$\{a, b\},\{c\}$ & 0.7& 0.7                & 0.4             \\ \hline
$\{a, b, c\}$ & 0.8&0.8 & 0.8              \\ \hline
\end{tabular}
\caption{The expected errors using uniform federation of players in each coalition when players $a, b, c$ each have 25 samples, with parameters $\mue = 10, \var =1$.}
\label{tab:3big}
\vspace{-4mm}
\end{table}

To motivate our problem and clarify the types of analyses we will be exploring, we will first work through a simple mean estimation example. (The Github repository contains numerical calculations and full tables for this section.) Calculating the error each player can expect requires two parameters: $\mue$, which reflects the average error each player experiences when sampling data from its own personal distribution, and $\var$, which reflects the average variance in the true parameters between players. In this section we will use $\mue = 10, \var = 1$, but will discuss later how to handle when they may be imperfectly known.

First, we will analyze uniform federation, with three players, $a, b,$ and $c$. We will first assume that each player has 5 samples from their local data distribution: Table \ref{tab:allsmall} gives the error each player can expect in this situation. Arrangements equivalent up to renaming of players are omitted. Every player sees its error minimized in the \enquote{grand coalition} $\gcol$ where all three players are federating together. This implies that the only arrangement that is stable (core-stable or individually stable) is $\gcol$. 

Next, we assume that player $c$ increases the amount of samples it has from 5 to 25: Table \ref{tab:2small1big} demonstrates the error each player can expect in this situation. Here, the players have different preferences over which arrangement they would most prefer. The \enquote{small} players $a$ and $b$ would most prefer $\{a, b\} \{c\}$, whereas the \enquote{large} player $c$ would most prefer $\{a\}, \{b, c\}$ or (identically) $\{b\}, \{a, c\}$. However out of all of these coalition structures, only $\{a, b\}, \{c\}$ is stable (either core stable or individually stable). Note that $\{a\}, \{b, c\}$ is not stable because the coalition $\col =\{a, b\}$ is one where each player prefers $\col$ to its current situation. 

Thirdly, we will assume that all three players have 25 samples: this example is shown in Table \ref{tab:3big}. As in Table \ref{tab:allsmall}, the players have identical preferences. However, in this case, the players minimize their error by being alone. Overall, stability results from this example are part of a broader pattern we will analyze in later sections.  

Next, we will explore the two other methods of federation: coarse-grained and fine-grained. Both offer some degree of personalization, with varying levels of flexibility. 

Table \ref{tab:coarse} shows an example using coarse federation with four players: three have 30 samples each, and the fourth player has 300 samples. We assume the weights $\w_j$ are set optimally for each player. For conciseness, some columns and rows are omitted. Note that both player types get lower error in $\gcol$ than they would with local learning: that is, $\gcol$ is individually stable (stable against defections of any player alone). However, it is also clear that $\gcol$ is not core stable: in particular, the three small players would get lower error in $\{a, b, c\}$ than in $\gcol$. These results will be examined theoretically in later sections: with optimal weighting, coarse-grained federation will always have an individually stable $\gcol$ that is not necessarily core stable.  

\begin{table}[]
\centering
\begin{tabular}{|c|c|c|}
\hline
Coalition structure & $err_a(\cd)$ & $err_d(\cd)$\\ \hline
$\{a\},\{b\},\{c\}, \{d\}$ & 0.333& 0.0333  \\\hline 
$\{a, b, c\},\{d\}$ &0.278 & 0.0333             \\ \hline
$\{a, b, c, d\}$ & 0.280 & 0.0326             \\ \hline
\end{tabular}
\caption{The expected errors using optimal coarse-grained federation when players $a, b, c$ each have 30 samples, while player $d$ has 300 samples, with parameters $\mue = 10, \var =1$.}
\label{tab:coarse}
\vspace{-4mm}
\end{table}

Finally, we examine fine-grained federation. Table \ref{tab:fine} analyzes the same case as coarse-grained federation previously, but with optimally-weighted fine-grained federation. The full error table shows that $\gcol$ \emph{is} core stable because each player minimizes their error in this arrangement. This result will hold theoretically: when optimal fine-grained federation is used, $\gcol$ always minimizes error for every player and is thus core stable. 

\begin{table}[]
\centering
\begin{tabular}{|c|c|c|}
\hline
Coalition structure & $err_a(\cd)$ & $err_d(\cd)$\\ \hline
$\{a\},\{b\},\{c\}, \{d\}$ & 0.333& 0.0333  \\\hline 
$\{a, b, c\},\{d\}$ &0.278 & 0.0333             \\ \hline
$\{a, b, c, d\}$ & 0.269 & 0.0325             \\ \hline
\end{tabular}
\caption{The expected errors using optimal fine-grained federation when players $a, b, c$ each have 30 samples, while player $d$ has 300 samples, with parameters $\mue = 10, \var =1$.}
\label{tab:fine}
\vspace{-4mm}
\end{table}

In later sections we will give theoretical results that explain this example more fully, but understanding the core-stable partitions here will help to build intuition for more general results.  

\section{Model and assumptions}
\ifthenelse{\equal{\camera}{1}}{}{\subsection{Model and technical assumptions}}
This section introduces our model. We assume that there is a fixed set of $[\nplayer]$ players, and player $j$ has a fixed number of samples, $\ndraw_j$. Though the number of samples is fixed, it is possible to analyze a varying number of samples by investigating all games involving the relevant number of samples. Each player draws their true parameters i.i.d. (independent and identically distributed) $(\mean_j, \err_j) \sim \Theta$. $\err_j$ represents the amount of noise in the sampling process for a given player. 

In the case of mean estimation, $\mean_j$ is a scalar representing the true mean of player $j$. Player $j$ draws samples i.i.d. from its true distribution: $\Ymf \sim \sampledist_j(\mean_j, \err_j)$. Samples are drawn with variance $\err_j$ around the true mean of the distribution. 

In the case of linear regression, $\param_j$ is a $\dimval$-dimensional vector representing the coefficients on the true classification function, which is also assumed to be linear. Each player draws $\ndraw_j$ input datapoints from their own input distribution $\X_j \sim \xdist{j}$ such that $\expxlin{j}[\x^T\x] = \cross{j}$. They then noisily observe the outputs, drawing values i.i.d. $\Y_j \sim \sampledist_j(\X_j^T \bm{\param}_j, \err_j)$, where $\err_j$ again denotes the variance of how samples are drawn around the true mean. 

There are three methods of federation. In uniform federation, a single model is produced for all members of the federating coalition: 
$$\hat{\param}^f = \frac{1}{\sum_{i=1}^{\nplayer}\ndraw_i} \sum_{i=1}^{\nplayer}\hat{\param}_i \cd \ndraw_i$$
In coarse-grained federation, each player has a parameter $\w_j$ that it uses to weight the global model with its own local model, producing an averaged model: 
$$\hat{\mean}^{\w}_j= \w_j \cd \hat{\mean}_j + (1-\w_j) \cd \frac{1}{\total}\sum_{i=1}^{\nplayer}\hat{\mean}_i \cd \ndraw_i$$
for $\w_j \in [0,1]$. Note that $\w_j=0$ corresponds to unweighted federated learning and $\w_j=1$ corresponds to pure local learning.
Finally, with fine-grained federation, each player $j$ as a vector of weights $\bm{\vm}_j$ that they use to weight every other player's contribution to their estimate: 
$$\hat{\mean}_j^{\vm} = \sum_{i=1}^{\nplayer}\vm_{ji}\mean_i$$
for $\sum_{i=1}^{\nplayer}\vm_{ji} = 1$. Note that we can recover the $\w$ weighting case with $\vm_{jj} = \w + \frac{(1-\w) \cd \ndraw_j}{\total}$ and $\vm_{ji} = (1-\w) \cd \frac{\ndraw_i}{\total}$.
Coarse-grained and fine-grained federation each have player-specific parameters ($\w, \vm$) that can be tuned. When those parameters are set optimally for the given player, we refer to the models as \enquote{optimal} coarse-grained or fine-grained federation. We will prove in later sections how to calculate optimal weights. 

We denote $\mue = \expparam{i}[\err_i]$: the expectation of the error parameter. In the mean estimation case, $\var = Var(\mean_i)$ represents the variance around the mean. In the linear regression case, $\var_d = Var(\param^d_i)$ for $d \in [\dimval]$.

We assume that each player knows how many samples it has access to. It may or may not have access to the data itself, but it does not know how its values (or its parameters) differ from the mean. For example, it does not know if the data it has is unusually noisy or if its true mean lies far from the true mean of other players. 

All of the stability analysis results depend on the parameters $\mue$ and $\var$. However, the reliance is fairly weak: often the player only needs to whether the number of samples they have $\ndraw_j$ is larger or smaller than the ratio $\frac{\mue}{\var}$.

Much of this paper analyzes the stability of coalition structures. Analyzing stability could be relevant because players can actually move between coalitions. However, even if players aren't able to actually move, analyzing the stability of a coalition tells us something about its optimality for each set of players. 

\ifthenelse{\equal{\camera}{1}}{}{
\subsection{Normative assumptions}
This paper is primarily descriptive: it aims to model a phenomenon in the world, not to say whether that phenomenon is good or bad. For example, it could be that society as a whole values situations where many players federate together and might wish to require players to do so, regardless of whether this minimizes their error. It might be the case that society prefers all players, regardless of how many samples they have access to, have roughly similar error rates. Our use of the expected mean squared error is also worth reflecting on: it assumes that over- and under-estimates are equally costly and that larger mis-estimates are more costly. In a more subtle point, we are taking the expected MSE over parameter draws $\expparam{i}$. A player with a true mean that happens to fall far from the mean might experience a much higher error than its expected MSE. 

In the entirely of this paper, we are taking as fixed the requirement that data not be shared, either for privacy or technical capability reasons, and so are implicitly valuing that requirement more than the desire for lower error. We are also assuming that the problem at hand is completely encompassed by the machine learning task, which might omit the fact that non-machine learning solutions may be better suited. It also may be the fact that technical requirements other than error rate are more important: for example, the desire to balance the amount of computation done by each agent. }

\section{Expected error results}

This paper's first contribution is to derive exact expected values for the MSE of players under different situations. The fact that these values are exact allows us to precisely reason about each player's incentives in later sections. We will state the theorems here and provide the proofs in  \ifthenelse{\equal{\camera}{1}}{the full version \cite{donahue2020modelsharing}.}{Appendix \ref{app:experr}.}

The approach for this section was first to derive expected MSE values for the most general case and then derive values for other cases as corollaries. The most general case is linear regression with fine-trained federation. First, we note that we can derive coarse-grained or uniform federation by setting the $\vm_{ji}$ weights to the appropriate values. Next, we note that mean estimation is a special case of linear regression. For intuition, consider a model where a player draws an $x$ value that is deterministically 1, then multiplies it by an unknown single parameter $\mean_j$, then then takes a measurement $y$ of this mean with noise $\err_j$. This corresponds exactly to the mean estimation case, where a player has a true mean $\mean_j$ and observes $y$ as a sample, with noise $\err_j$. We can use this representation to simplify the error terms, with more details given in \ifthenelse{\equal{\camera}{1}}{the full version \cite{donahue2020modelsharing}.}{Appendix \ref{app:experr}.}

First, we give the expected MSE for local estimation: 
\begin{restatable}{theorem}{linloc}
\label{linloc}
For linear regression, the expected MSE of local estimation for a player with $\ndraw_j$ samples is
$$\mue\cd \text{tr}\br{\cross{j}\expXlin{j}\br{\p{\X_j^T\X_j}^{-1}}}$$
If the distribution of input values $\xdist{j}$ is a $\dimval$-dimensional multivariate normal distribution with 0 mean, then, the expected MSE of local estimation can be simplified to: 
$$\frac{\mue}{\ndraw_j - \dimval -1}\dimval$$
In the case of mean estimation, the error term can be simplified to: 
$$\frac{\mue}{\ndraw_j}$$
\end{restatable}

Next, we calculate the expected MSE for fine-grained federation: 
\begin{restatable}{theorem}{linfedv}
\label{linfedv}
For linear regression with fine-grained federation, the expected MSE of federated estimation for a player with $\ndraw_j$ samples is: 
$$L_j + \p{\sum_{i\ne j}\vm_{ji}^2 + \p{\sum_{i\ne j}\vm_{ji}}^2}\cd \sum_{d=1}^{\dimval}\expxlin{j}[(\x^d)^2] \cd \var_d$$
where $L_j$ is equal to:
$$\mue\sum_{i=1}^{\nplayer}\vm_{ji}^2\cd \text{tr}[\cross{j}\expdata{i}\br{(\X_i^T\X_i)^{-1}}$$
If the distribution of input values $\xdist{i}$ is a $\dimval$-dimensional multivariate normal distribution with 0 mean, this can be simplified to:
$$\mue\sum_{i=1}^{\nplayer}\vm_{ji}^2\cd \frac{\dimval}{\ndraw_i - \dimval -1}$$
In the case of mean estimation, the entire error term can be simplified to: 
$$\mue\sum_{i=1}^{\nplayer}\vm_{ji}^2\cd \frac{1}{\ndraw_i} + \p{\sum_{i\ne j}\vm_{ji}^2 + \p{\sum_{i\ne j}\vm_{ji}}^2}\cd \var$$
\end{restatable}
Finally, we derive as corollaries the expected MSE for the uniform federation and the coarse-grained case. 

\begin{restatable}{corollary}{linfed}
\label{linfed}
For uniform linear regression, the expected MSE of federated estimation for a player with $\ndraw_j$ samples is: 
$$L_j + \frac{\sum_{i\ne j}\ndraw_i^2 + (\total - \ndraw_j)^2}{\total^2}\sum_{d=1}^{\dimval}\expxlin{j}[(\x^d)^2] \cd \var_d$$
where $L_j$ is equal to:
$$\mue\sum_{i=1}^{\nplayer}\frac{\ndraw_i^2}{\total^2}\text{tr}[\cross{j}\expdata{i}\br{(\X_i^T\X_i)^{-1}}$$
or, if the distribution of input values $\xdist{i}$ is a $\dimval$-dimensional multivariate normal distribution with 0 mean, can be simplified to
$$\mue\sum_{i=1}^{\nplayer}\frac{\ndraw_i^2}{\total^2}\frac{\dimval}{\ndraw_i - \dimval -1}$$
In the case of mean estimation, the entire error term can be simplified to: 
$$\frac{\mue}{\total} + \frac{ \sum_{i\ne j}\ndraw_i^2  +(\total - \ndraw_j)^2}{\total^2}\var$$
where $\total = \sum_{i=1}^{\nplayer}\ndraw_i$. 
\end{restatable}

\begin{restatable}{corollary}{linfedw}
\label{linfedw}
For coarse-grained linear regression, the expected MSE of federated estimation for a player with $\ndraw_j$ samples is: 
$$L_j + (1-\w)^2 \cd \frac{\sum_{i\ne j}\ndraw_i^2 + (\total - \ndraw_j)^2}{\total^2}\sum_{d=1}^{\dimval}\expxlin{j}[(\x^d)^2] \cd \var_d$$
where $L_j$ is equal to:
$$\mue \cd (1-\w)^2 \cd \sum_{i=1}^{\nplayer}\frac{\ndraw_i^2}{\total^2}\text{tr}[\cross{j}\expdata{i}\br{(\X_i^T\X_i)^{-1}} +$$
$$\mue \p{\w^2 + 2\frac{(1-\w) \w \cd  \ndraw_j}{\total}} \cd \text{tr}[\cross{j}\expdata{i}\br{(\X_j^T\X_j)^{-1}}$$
or, if the distribution of input values $\xdist{i}$ is a $\dimval$-dimensional multivariate normal distribution with 0 mean, can be simplified to
$$\mue \cd (1-\w)^2 \cd \sum_{i=1}^{\nplayer}\frac{\ndraw_i^2}{\total^2}\frac{\dimval}{\ndraw_i - \dimval -1}$$
$$+ \mue \cd \p{\w^2 + 2 \cd \frac{(1-\w) \cd \w \cd \ndraw_j}{\total}} \cd \frac{\dimval}{\ndraw_i - \dimval -1}$$
In the case of mean estimation, the entire error term can be simplified to: 
$$\mue \p{\frac{\w^2}{\ndraw_j} + \frac{1-\w^2}{\total}} + \frac{ \sum_{i\ne j}\ndraw_i^2  +(\total - \ndraw_j)^2}{\total^2}\cd (1-\w)^2\var$$
where $\total = \sum_{i=1}^{\nplayer}\ndraw_i$. 
\end{restatable}

The exact MSE for linear regression follows a very similar form to that for mean estimation. In all cases, the bias component (the term involving $\var_d$) is in the exact same form and could be directly modified to mean estimation by using $(\var)' = \sum_{d=1}^{\dimval}\expxlin{j}[(\x^d)^2] \cd \var_d$. The variance component (the term involving $\mue$) fits the exact form of mean estimation in the limit where $\ndraw_j >> \dimval$. In this case, the error can be modified to fit mean estimation by using $\mue' = \dimval \cd \mue$. This approximation is good when there are many more samples than the dimension of the linear regression problem under investigation: for most cases of model fitting, this assumption is reasonable. 

For the rest of the paper, we will use the $\ndraw_j >> \dimval$ assumption: consequentially, \emph{all of our results apply equally to linear regression and mean estimation}.

\section{Uniform federation: coalition formation}\label{sec:unadj}
In this section, we analyze the stability of coalition structures in the case that uniform federation is used. We consider two cases: 1) where all players have the same number of datapoints $\ndraw$ and 2) where all players have either a \enquote{small} or \enquote{large} number of points. We will use $\alone$ to refer to the coalition partition where all players are alone and $\gcol$ to refer to the grand coalition. Proofs from this section are given in \ifthenelse{\equal{\camera}{1}}{the full version \cite{donahue2020modelsharing}.}{Appendix \ref{app:coalition}.}

\subsection{All players have the same number of samples}
In this case, the analysis simplifies greatly: 
\begin{lemma}
\label{nsameunadj}
If all players have the same number of samples $\ndraw$, then: \begin{itemize}
    \item If $\ndraw < \frac{\mue}{\var}$, players minimize their error in $\gcol$.   
    \item If $\ndraw > \frac{\mue}{\var}$, players minimize their error in $\alone$.  
    \item If $\ndraw = \frac{\mue}{\var}$, players are indifferent between any arrangement of players. 
\end{itemize}
\end{lemma}
\begin{proof}
In the case that all players have the same number of samples, we can use $\ndraw_i=\ndraw$ to simplify the error term: 
$$\frac{\mue}{\nplayer \cd \ndraw} + \var \frac{\nplayer-1}{\nplayer}$$
In order to see whether players would prefer a larger group (higher \nplayer) or a smaller group (smaller \nplayer), we take the derivative of the error with respect to $\nplayer$:
$$-\frac{\mue}{\nplayer^2 \cd \ndraw} + \frac{\var}{\nplayer^2}= \frac{\var \cd \ndraw - \mue}{\ndraw\cd \nplayer^2}$$
This is positive when $\ndraw > \frac{\mue}{\ndraw}$: a player gets higher error the more players it is federating with. This is negative when $\ndraw < \frac{\mue}{\var}$: a player gets lower error the more players it is federating with. This is 0 when $\ndraw = \frac{\mue}{\var}$, which implies players should be indifferent between different arrangements. Plugging in for $\ndraw = \frac{\mue}{\var}$ in the error equation gives
$\frac{\mue\cd \var}{\nplayer \cd \mue}  + \var \frac{\nplayer -1}{\nplayer} = \var$
which is equivalent to the error a player would get alone: 
$\frac{\mue}{\ndraw} =  \frac{\mue\cd \var}{\mue} = \var$. 
\end{proof}
As a corollary, we can classify the core stable arrangements cleanly: 
\begin{restatable}{corollary}{stabsame}
\label{stabsame}
For uniform federation, if all players have the same number of samples $\ndraw$, then: \begin{itemize}
    \item If $\ndraw < \frac{\mue}{\var}$, $\gcol$ is the only partition that is core-stable.  
    \item If $\ndraw > \frac{\mue}{\var}$, $\alone$ is the only partition that is core-stable.
    \item If $\ndraw = \frac{\mue}{\var}$, any arrangement of players is core-stable. 
\end{itemize}
\end{restatable}

\subsection{Small \& large player case}
In this section, we add another layer of depth by allowing players to come in one of two \enquote{sizes}. \enquote{Small} players have $\ns$ samples and \enquote{large} ones have $\nl$ samples, with $\ns < \nl$. We demonstrate that versions of the game in this pattern always have a stable partition by constructively producing an element that is stable. Note that this is \emph{not} true in general of hedonic games. As discussed in \citet{BOGOMOLNAIA2002201}, there are multiple instances where a game might have no stable partition. 

To characterize this space, we divide it into cases depending on the relative size of $\ns, \nl$. We will use the notation $\pi(\s, \el)$ to denote a coalition with $\s$ small players and $\el$ large players, out of a total of $\Sv$ and $\Lv$ present. We will use $\pi(\s_1, \el_1) \cg_{\si} \pi(\s_2, \el_2)$ to mean that the small players prefer coalition $\pi(\s_1, \el_1)$ to $ \pi(\s_2, \el_2)$ and $\pi(\s_1, \el_1) \cg_{\li} \pi(\s_2, \el_2)$ to mean the same preference, but for large players. 

\subsubsection{Case 1: $\ns, \nl \geq \frac{\mue}{\var}$}
The first case is when $\ns$ is large: it turns out that each player minimizes their error by using local learning,  which means that $\alone$ is in the core. The lemma below is more general than the small/large case, but implies that when $\ns> \frac{\mue}{\var}$, $\alone$ is the only element in the core and when $\ns = \frac{\mue}{\var}$ then any arrangement where the large players are alone are is in the core. 

\begin{restatable}{lemma}{allbigeq}
\label{allbigeq}
For uniform federation, if $\ndraw_i > \frac{\mue}{\var}$ for all $i\in [\nplayer]$, then $\alone$ is the unique element in the core. \\
If $\ndraw_i \geq \frac{\mue}{\var}$ for all $i\in [\nplayer]$, with $\ndraw_k > \frac{\mue}{\var}$ for at least one player $k$, then any arrangement where the players with samples $\ndraw_k > \frac{\mue}{\var}$ are alone is in the core. 
\end{restatable}

\subsubsection{Case 2: $\ns, \nl \leq \frac{\mue}{\var}$}
Next, we consider the case where both the small and large players have a relatively small number of samples. In this situation, it turns out that the grand coalition is core stable.

\begin{restatable}{theorem}{smallcore}
\label{smallcore}
For uniform federation, if $\nl\leq \frac{\mue}{\var}$ and $\ns < \nl$, then the grand coalition $\gcol$ is core stable.  
\end{restatable}

\subsubsection{Case 3: $\ns<\frac{\mue}{\var}$, $\nl > \frac{\mue}{\var}$}

Finally, we consider the case where the small players have a number of samples below the $\frac{\mue}{\var}$ boundary, while the large players have a number of samples above this threshold. 

\begin{restatable}{theorem}{indivstab}
\label{indivstab}
Assume uniform federation with $n_{\ell} > \frac{\mu_e}{\sigma^2}$. Then, there exists an arrangement of small and large players that is individually stable and a computationally efficient algorithm to calculate it. 
\end{restatable}
The proof of Theorem \ref{indivstab} is constructive: it gives an exact arrangement that is individually stable. One natural question is whether this arrangement is also core stable. The answer to this question is \enquote{no}: we show that this arrangement can fail to be core stable. This avenue is explained more in \ifthenelse{\equal{\camera}{1}}{the full version \cite{donahue2020modelsharing}.}{Appendix \ref{app:coalition}.}

\section{Coarse-grained federation}
In this section, we analyze coarse-grained federation. As a reminder, in this situation, each player has a parameter $\w_j$ that it uses to weight the global model with its own local model. 
$$\hat{\mean}^{\w}_j= \w_j \cd \hat{\mean}_j + (1-\w_j) \cd \frac{1}{\total}\sum_{i=1}^{\nplayer}\hat{\mean}_i \cd \ndraw_i$$
for $\w_j \in [0,1]$. All proofs from this section are given in \ifthenelse{\equal{\camera}{1}}{the full version \cite{donahue2020modelsharing}.}{Appendix \ref{app:wsupp}.}

Note that the $\w_j$ value is a parameter that each player can set independently. The lemma below analyzes the optimal value of $\w_j$ and tells us that each player would prefer federation, in some form, to being alone. 

\begin{restatable}{lemma}{wfedminval}
\label{wfedminval}
For coarse-grained federation, the minimum error is always achieved when $\w_j<1$, implying that federation is always preferable to local learning. 
\end{restatable}
\begin{corollary}
For coarse-grained federation, when $\w_j$ is set optimally, the grand coalition $\gcol$ is always individually stable. 
\end{corollary}

Specifically, this means that no player wishes to unilaterally deviate from $\gcol$. However, this does \emph{not} mean that each player prefers the grand coalition $\gcol$ to some other federating coalition. For example, refer to Section \ref{motivate} for an example where the grand coalition $\gcol$ is not core stable.

In the rest of this section, we will analyze the stability of coalition structures in the that the $\w$ parameters are set optimally (optimal coarse-grained federation). First, we will find it useful to get the closed-form value for expected MSE of a player using optimal coarse-grained federation: 

\begin{restatable}{lemma}{wbesterr}
\label{wbesterr}
A player using coarse-grained federation parameter has expected MSE: 
$$\frac{\mue \cd (\total - \ndraw_j) + (\sum_{i\ne j}\ndraw_i^2 + (\total-\ndraw_j)^2) \cd \var}{(\total - \ndraw_j) \cd \total + \ndraw_j \cd (\sum_{i\ne j}\ndraw_i^2 + (\total-\ndraw_j)^2) \cd \frac{\var}{\mue}}$$
where $\total = \sum_{i=1}^{\nplayer} \ndraw_i$. 
\end{restatable}

\subsection{All players have the same number of samples}
Lemma \ref{nsamew} is the analog to Lemma \ref{nsameunadj} in the previous section. Here, the results differ: with optimal coarse-grained federation, the grand coalition $\gcol$ is \emph{always} the only stable arrangement, no matter how small or large $\ndraw$ is relative to $\frac{\mue}{\var}$. 

\begin{lemma}
\label{nsamew}
For mean estimation with coarse-grained federation, if $\ndraw_j = \ndraw$, then $\gcol$ is the only element in the core. 
\end{lemma}
\begin{proof}
Using the error term derived in Lemma \ref{wbesterr}, plugging in for $\ndraw_i = \ndraw$ and simplifying gives: 
$$\frac{\frac{\mue^2}{\ndraw \cd \nplayer}+ \mue \cd\var }{\mue + \ndraw  \cd \var} $$
As $\nplayer$ increases, the error (numerator) decreases always - so $\gcol$ is where each player minimizes their error and is thus core stable.
\end{proof}

\subsection{Small \& large player case}
In this subsection, we similarly extend results for the \enquote{small} and \enquote{large} case that was introduced in the previous section. The analysis turns out to be much simpler than in the uniform federation case, and also produce stronger results: strict core stability, rather than individual stability.

\begin{restatable}{theorem}{wslcore}
\label{wslcore}
If optimal coarse-grained federation is used, then: 
\begin{itemize}
    \item If $\gcol \cleq_{\si} \pi(\Sv, 0)$ (small player weakly prefers $\pi(\Sv, 0)$), then $\{\pi(\Sv, 0), \pi(0, \Lv)\}$ is strictly core stable. 
    \item If $\gcol \cg_{\si} \pi(\Sv, 0)$ (small player strictly prefers $\gcol$), then $\gcol$ is strictly core stable. 
\end{itemize}
\end{restatable}

\section{Fine-grained federation}
In this section, we analyze fine-grained federation. As a reminder, with this method, each player $j$ as a vector of weights $\bm{\vm}_j$ that they use to weight every other player's contribution to their estimate. 
$$\hat{\mean}_j^{\vm} = \sum_{i=1}^{\nplayer}\vm_{ji}\mean_i$$
for $\sum_{i=1}^{\nplayer}\vm_{ji} = 1$. 

We calculate the optimal $\bm{\vm}$ weights for player $j$'s error. 

\begin{restatable}{lemma}{vfedminval}
\label{vfedminval}
Define $V_i = \var + \frac{\mue}{\ndraw_i}$. Then, the value of $\{\vm_{ji}\}$ that minimizes player $j$'s error is: 
$$\vm_{jj} = \frac{1 + \var \sum_{i\ne j}\frac{1}{V_i}}{1 + V_j\sum_{i\ne j}\frac{1}{V_i}}$$
$$\vm_{jk} = \frac{1}{V_k}\cd \frac{V_j-\var}{1 + V_j \sum_{i\ne j}\frac{1}{V_i}} \quad k\ne j$$
\end{restatable}
The proof of this lemma is given in \ifthenelse{\equal{\camera}{1}}{the full version \cite{donahue2020modelsharing}.}{Appendix \ref{app:vsupp}.}

From this analysis, a few properties become clear. To start with, $\vm_{jj}$ and $\vm_{jk}$ are always strictly between 0 and 1. This implies the following lemma: 

\begin{corollary}
With optimal fine-grained federation, $\gcol$ is optimal for each player.
\end{corollary} 
\begin{proof}
Suppose by contradiction that some other coalition $\pi'$ gave player $j$ a lower error. WLOG, assume this coalition omitted player $k$. In this case, the $\vm$ weights for $\pi'$ can be represented as a length $\nplayer$ vector with $0$ in the $k$th entry. However, set of weights is achievable in $\gcol$: it is always an option to set a player's coefficient $\vm_{jk}$ equal to 0. This contradicts the use of $\bm{\vm}_j$ as an optimal weighting, so it cannot be the case that any player gets lower error in a different coalition. 
\end{proof}

Similarly, the fact that $\gcol$ is optimal for every player implies that it is in the core, and that it is the only element in the core.

\section{Conclusions and future directions}

In this work, we have drawn a connection between a simple model of federated learning and the game theoretic tool of hedonic games. We used this tool to examine stable partitions of the space for two variants of the game. In service of this analysis, we computed exact error values for mean estimation and linear regression, as well as for three different variations of federation. 

We believe that this framework is a simple and useful tool for analyzing the incentives of multiple self-interested agents in a learning environment. There are many fascinating extensions. For example, completely characterizing the core (including whether it is always non-empty) in the case of arbitrary number of samples $\{\ndraw_i\}$ is an obvious area of investigation. Besides this, it could be interesting to compute exact or approximate error values for cases beyond mean estimation and linear regression. 

\ifthenelse{\equal{\camera}{1}}{
\paragraph{\bf Normative assumptions (ethics statement)}
This paper is primarily descriptive: it aims to model a phenomenon in the world, not to say whether that phenomenon is good or bad. For example, it could be that society as a whole values situations where many players federate together and might wish to require players to do so, regardless of whether this minimizes their error. It might be the case that society prefers all players, regardless of how many samples they have access to, have roughly similar error rates. Our use of the expected mean squared error is also worth reflecting on: it assumes that over- and under-estimates are equally costly and that larger mis-estimates are more costly. In a more subtle point, we are taking the expected MSE over parameter draws $\expparam{i}$. A player with a true mean that happens to fall far from the mean might experience a much higher error than its expected MSE. 

In the entirely of this paper, we are taking as fixed the requirement that data not be shared, either for privacy or technical capability reasons, and so are implicitly valuing that requirement more than the desire for lower error. We are also assuming that the problem at hand is completely encompassed by the machine learning task, which might omit the fact that non-machine learning solutions may be better suited. It also may be the fact that technical requirements other than error rate are more important: for example, the desire to balance the amount of computation done by each agent. }{}

\subsection*{Acknowledgments}
This work was supported in part by a Simons Investigator Award, a Vannevar Bush Faculty Fellowship, a MURI grant, AFOSR grant FA9550-19-1-0183, grants from the ARO and the MacArthur Foundation, and NSF grant DGE-1650441. We are grateful to A. F. Cooper, Thodoris Lykouris, Hakim Weathersppon, and the AI in Policy and Practice working group at Cornell for invaluable discussions. In particular, we thank A.F. Cooper for discussions around normative assumptions. Finally, we are grateful to Katy Blumer for discussions around code in the Github repository.

\bibliography{biblio.bib}
\newpage
\appendix
\ifthenelse{\equal{\camera}{1}}{\clearpage }{}
\section{Relationship to other approaches}\label{app:relate}
This section contains a high-level summary of similar approaches and how they relate to ours. Throughout we assume the goal is to estimate some unknown $\mean_j$ given samples drawn $Y_i \sim D(\theta_j)$. 

A \textbf{frequentist approach}  would take $\mean_i$ to be a constant that would be estimated by the average of the given samples $\frac{1}{\ndraw_j}\sum_{i=1}Y_i$. 

A \textbf{hierarchical Bayesian} estimator assumes data is generated in the following way: data is drawn $Y_i \sim D(Y \vert \mean_i)$. The parameter $\mean_i$ is drawn $\mean_i \sim \gendist_i(\mean \vert \lambda_i)$, where hyperparameter $\lambda_i$ is drawn from known distribution $p(\lambda)$. Given some data, the parameter $\mean_i$ can be estimated as follows
$$p(\mean_i \vert Y_i) = \frac{p(Y_i \vert \mean_i) p(\mean_i)}{p(Y_i)} = \frac{p(Y_i \vert \mean_i)}{p(Y_i)}\int p(\mean_i \vert \lambda_i) p(\lambda_i) d\lambda$$

\textbf{Parametric empirical Bayes} \cite{morris} \cite{casella1992illustrating} is frequently described as an intermediate between these two viewpoints. Similar to the hierarchical Bayesian viewpoint, it assumes data is drawn $Y_i \sim D(Y \vert \mean_i)$, with parameter $\mean_i$ is drawn $\mean_i \sim \gendist_i(\mean \vert \lambda_i)$. However, it differs in that it estimates $\lambda_i$ based on the data, producing $\hat \lambda_i$. This estimate of the hyperparameter is used, along with the data, to estimate $\mean_i$. 

A related example is the \textbf{James-Stein estimator}  \cite{efron1977stein}. The estimator assumes the following process: each of $m$ players draws a single sample from a normal distribution with variance $s^2$. 
$$Y_i \sim \mathcal{N}(\mean_i, s^2)$$
This is different from the empirical Bayes or Bayes case in that it is assumed that the means $\mean_i$ are completely unrelated to each other. Nevertheless, it has been demonstrated that the James-Stein estimator: 
$$\hat \mean_{JS} = \p{1 - \frac{(m-2) \cd s^2}{\norm{\bf{Y}}^2}} \bf{Y}$$
has lower expected MSE than simply using the drawn parameters $Y_i$. In the case that the variance $s^2$ is not known perfectly, it can be estimated as $\hat{s^2}$ using entire vector of data $\bf{Y}$. 

\textbf{Our method} is similar at a high level to empirical Bayes: we assume each player draws data from a personal distribution governed by $\mean_i$ and that the $\mean_i$ terms are in turn drawn from some distribution $\gendist$. However, one key difference is that all three methods discussed above assume knowledge of the distributions generating the data, or at least which family they are drawn from. For example, the James Stein estimator assumes a normal distribution: variants of it exist for different distributions, but not a version that works for all distributions. Similarly, a hierarchical Bayes or empirical Bayes viewpoint would require knowledge of the $D, \gendist, p$ distributions. In our approach, we do not assume that we know the form of these generating distributions, only some summary values summary statistics (mean and variance) of the distribution. 

It is entirely possible that other approaches, especially those that assume knowledge of the generating distribution, will out-perform our approach in terms of the error guarantees they can provide. Our distribution-free approach allows it to be implemented in a broader range of situations. 

Additionally, our approach is restricted to linear combinations of estimators such as $\hat{\mean}^f = \w \cd \hat{\mean}_j + (1-\w) \sum_{i=1}^{\nplayer}\hat{\mean}_i$. It is possible that a method outside this situation, for example, something like $\hat{\mean}^f = x \cd \hat{\mean}_j + y \sum_{i=1}^{\nplayer}\hat{\mean}_i^2$, or something like the non-linear James Stein estimator, would produce better estimates.

\section{Expected error proofs}\label{app:experr}

For convenience, we will restate the model setup for the most general case of linear regression. We assume that each player $j \in [\nplayer]$ draws parameters $(\bm{\param}_j, \err_j) \sim \gendist$, where $\bm{\param}_j$ is a length $\dimval$ vector and $\err_j$ is a scalar-valued variance parameter. The $d$th entry in the vector is $\param_j^d$, and $Var(\param_j^d) = \var_d$. We assume that each value $\param_j$ is drawn independently of the others. The main result of this section will assume that each dimension is drawn independently, for example that $\param_j^l$ is independent of $\param_j^k$, for $k\ne l$, but we will demonstrate how this can be relaxed. Each player draws $\ndraw_j$ input data points from their own input distribution, $\X_j \sim \xdist{j}$ such that $\expxlin{j}[\x^T\x] = \cross{j}$. They then noisily observes the outputs, drawing $\Y_j \sim \sampledist_j(\X_j^T \bm{\param}_j, \err_j)$. We use $\errval_j$ to denote the length $\dimval$ vector of errors so that $\Y_j = \X_j^T\bm{\param}_j + \bm{\errval}_j$. Each player uses ordinary least squares (OLS) to compute estimates of their parameters, which requires that $\X_j^T\X$ is invertible. This happens when the columns of $\X$ are linearly independent. We will require that $\cross{j}$ be such that $\X_j^T\X$ is invertible with probability 1. This rules out cases where one dimension is a deterministic function of another, for example. Using the below OLS calculation gives local estimation:
$$\hat{\bm{\param}_j} = (\X^T_j\X_j)^{-1}\Y_j = (\X^T_j\X_j)^{-1}(\X_j \bm{\param}_j + \errval_j)$$
It is worth pausing briefly to note why mean estimation is a special case of linear regression. Consider the case where the distribution $\xdist{j}$ is deterministically 1, meaning that $\X_j$ is a vector of 1s of length $\ndraw_j$. $(\X_j^T\X_j)$ is always invertible, as $(\X_j^T\X_j)^{-1} = \ndraw_j^{-1}$. Each $\X_j$ is multiplied by an unknown single parameter $\mean_j$, which each player is attempting to learn. 

Besides local estimation, there are multiple federation possibilities. Uniform federation is given by: 
$$\hat{\bm{\param}}_j^f = \frac{1}{N} \sum_{i=1}^{\nplayer}\hat{\bm{\param}}_i \cd \ndraw_i$$
Coarse-grained federation is given by: 
$$\hat{\mean}^{\w}_j= \w_j \cd \hat{\bm{\param}}_j + (1-\w_j) \cd \frac{1}{\total}\sum_{i=1}^{\nplayer}\hat{\bm{\param}}_i \cd \ndraw_i$$
for $\w_j \in [0,1]$. Finally, fine-grained federation is given by: 
$$\hat{\mean}_j^{\vm} = \sum_{i=1}^{\nplayer}\vm_{ji}\mean_i$$
for $\sum_{i=1}^{\nplayer}\vm_{ji} = 1$. Note that fine-grained federation is the most general case of federation. It is possible to derive coarse-grained federation, uniform federation, or local estimation by appropriately setting the $\vm$ weights. In this section, we will first derive the expected error for the fine-grained federation, linear regression case, and get expected error results for other cases as corollaries of this result. 

The expected error produced by a set of estimates $\hat{\bm{\param}}$ is determined by the expectation of the following quantity. 
$$(\x^T\hat{\bm{\param}} - \x^T\bm{\param}_j)^2$$
Here, the expectation is taken over four sources of randomness. 
\begin{enumerate}
    \item $\expparam{i}$: Drawing parameters $\bm{\param}_j, \err_j$ for player $j$'s distribution from $\gendist$.
    \item $\expXlin{i}$: Drawing the training dataset $\X_i$ from data distribution $\xdist{i}$
    \item $\expYlin{i}$: Drawing labels for the training dataset $\X_i$ from the distribution $\sampledist_{i}(\X_i^T\bm{\param}_i, \err_i)$. 
    \item $\expxlin{j}$: Drawing a new test point $\x$ from the data distribution $\xdist{j}$
\end{enumerate}
$(\x^T\hat{\bm{\param}} - \x^T\bm{\param}_j)^2$ measures the expected error of a set of parameters at a particular point $\x$: when the expectation is taken over all $\x \sim \xdist{j}$, it represents the average error everywhere on the distribution. 
It might be not immediately clear, though, why $(\x^T\hat{\bm{\param}} - \x^T\bm{\param}_j)^2$ is the correct term to be considering. Other potential candidates might include: 
\begin{enumerate}
    \item $\norm{\hat{\bm{\param}}_j - \bm{\param}_j}^2$
    \item $(\x^T\hat{\bm{\param}}_j - y)^2$ for $y = \x^T\bm{\param}_j + \errvalmf_j$. 
\end{enumerate}
The first candidate measures the difference in estimated parameters; however, we assume that the objective of learning is to have low error on predicting future points, rather than solely estimate the parameters. The second candidate represents the error of predicting an instance as opposed to a mean value: it ends up simply producing an additive increase in our overall error term. To see this, note that we can write
$$(\x^T\hat{\bm{\param}}_j - y)^2 = (\x^T\hat{\bm{\param}}_j - \x^T\param_j + \x^T\param_j - y)^2$$
$$ = (\x^T\hat{\bm{\param}}_j - \x^T\param_j + \x^T\param_j - \x^T\param_j + \errvalmf_j)^2 $$
$$= (\x^T\hat{\bm{\param}}_j - \x^T\param_j + \errvalmf_j)^2$$
$$=(\x^T\hat{\bm{\param}}_j - \x^T\param_j)^2 + 2((\x^T\hat{\bm{\param}}_j - \x^T\param_j)\errvalmf_j + \errvalmf_j^2$$
The first term is the same error function we are considering. The middle term is 0 in expectation and the last term is $\mue$ in expectation, so this approach simply scales the error we were looking at by $\mue$.

\linloc*
Note: portions of \citet{abu2012learning} and \citet{PaquaySolutions}, especially problem 3.11, were helpful in formulating this approach. \citet{10.1093/mnrasl/slv190} and \citet{anderson1962introduction} were helpful in providing the connection to the Inverse Wishart.
\begin{proof}
First, note that: 
$$\x^T\bm{\param_j} - \x^T\hat{\bm{\param}}_j = \x^T\p{\bm{\param}_j -(\X_j^T\X_j)^{-1}\X_j^T\Y_j}$$
$$ =\x^T\p{\bm{\param}_j -(\X_j^T\X_j)^{-1}\X_j^T(\X_j\bm{\param}_j + \errval_j)} $$
$$ =\x^T\p{\bm{\param}_j -\bm{\param}_j - (\X_j^T\X_j)^{-1}\X_j^T\errval_j} $$
$$ =-\x^T(\X_j^T\X_j)^{-1}\X_j^T\errval_j $$
Then, 
$$(\x^T\bm{\param_j} - \x^T\hat{\bm{\param}}_j)^2$$
$$ = \errval_j^T\X_j(\X_j^T\X)^{-1}\x\x^T(\X_j^T\X_j)^{-1}\X_j^T\errval_j$$
To simplify, we note that the above quantity is a scalar. For a scalar, $a = \text{tr}(a)$, and for any matrix, $\text{tr}(AB) = \text{tr}(BA)$ through the cyclic property of the scalar. 
$$ = \text{tr}\br{\errval_j^T\X_j(\X_j^T\X)^{-1}\x\x^T(\X_j^T\X_j)^{-1}\X_j^T\errval_j}$$
$$ = \text{tr}\br{\x\x^T(\X_j^T\X_j)^{-1}\X_j^T\errval_j\errval_j^T\X_j(\X_j^T\X)^{-1}}$$
To evaluate, we start by applying the various expectations, noting that expectation and trace commute. Applying $\expetalin{j}$ to the term above allows us to rewrite it as: 
$$ = \text{tr}\br{\x\x^T(\X_j^T\X_j)^{-1}\X_j^TV\X_j(\X_j^T\X)^{-1}}$$
where
$$V = \expetalin{j}[\errval_j\errval_j^T]$$
$\errval_j\errval_j^T$ is an $\ndraw_j \times \ndraw_j$ matrix. The $l$th diagonal is $(\errval_j^l)^2$, which has expectation $\err_j$. Off diagonal entries have value $\errval_j^l \cd \errval_j^k$ for $\ell \ne k$. Because the errors for each data point are drawn independently and with 0 mean, the expectation of this is 0. $\expetalin{j}[\errval_j\errval_j^T]$ is a diagonal matrix with $\err_j$ along the diagonal: we can pull it out of the trace to obtain: 
$$=\err_j\text{tr}\br{\x\x^T(\X_j^T\X_j)^{-1}\X_j^T\X_j(\X_j^T\X)^{-1}}$$
$$=\err_j\text{tr}\br{\x\x^T(\X_j^T\X_j)^{-1}}$$
Taking the expectation over the drawn parameters gives: 
$$=\expparam{j}[\err_j\text{tr}\br{\x\x^T(\X_j^T\X_j)^{-1}}]$$
$$=\mue\text{tr}\br{\x\x^T(\X_j^T\X_j)^{-1}}$$
Taking the expectation over the test point $\x \sim \xdist{j}$ gives: 
$$=\mue\text{tr}\br{\expxlin{j}[\x\x^T](\X_j^T\X_j)^{-1}}$$
$$=\mue\text{tr}\br{\cross{j}(\X_j^T\X_j)^{-1}}$$
Finally, we take the expectation with respect to $\X_j \sim \xdist{j}$
$$=\mue\text{tr}\br{\cross{j}\expXlin{j}\br{\p{\X_j^T\X_j}^{-1}}}$$
Note that because the inverse and expectation do not commute, in general, we cannot simplify this without stronger assumptions. 

There is one other situation where a particular case of linear regression gives us simpler results. As mentioned in the statement of the lemma, in this case we assume that the distribution of input values $\xdist{j}$ is a 0-mean normal distribution with covariance matrix $\crossc{j}$. 

Note that, in general, $\crossc{j} \ne \cross{j} = \expxlin{j}[\x\x^T]$. $\expxlin{j}[\x\x^T]$ has, along the diagonals, $\expxlin{j}[\x_j^{d}\x_j^{d}]$, and on the off-diagonals, has $\expxlin{j}[\x_j^{l}\x_j^{k}]$. By contrast, the covariance matrix $\crossc{j}$ has the same term along the diagonals, but the off-diagonal term has $\expxlin{j}[\x_j^{l}\x_j^{k}] - \expxlin{j}[\x_j^{l}]\expxlin{j}[\x_j^{k}]$. In the case we are looking at, the distribution is 0 mean, so the off-diagonal terms match as well, and $\crossc{j} = \cross{j}$. 

If this is the case, then $(\X_j^T\X_j)$ is distributed according to a Wishart distribution with parameters $\ndraw_j$ and covariance $\crossc{j} = \cross{j}$ with dimension $\dimval$. Given this, $(\X_j^T\X_j)^{-1}$ is distributed according to an Inverse Wishart distribution with parameters $\ndraw_j$ and covariance $\crossc{j}^{-1} = \cross{j}^{-1}$ with dimension $\dimval$. 

The expectation of the inverse Wishart tells us that: 
$$\expXlin{j}\br{\p{\X_j^T\X_j}^{-1}} = \frac{1}{\ndraw_j-\dimval-1}\crossc{j}^{-1}$$
Using these results, we can directly calculate the desired expectation: 
$$\mue\text{tr}\br{\cross{j}\expXlin{j}\br{\p{\X_j^T\X_j}^{-1}}}$$
$$=\mue\text{tr}\br{\frac{1}{\ndraw_j - \dimval -1}\cross{j}\cross{j}^{-1}}$$
$$=\frac{\mue}{\ndraw_j - \dimval -1}\text{tr}\br{\cross{j}\cross{j}^{-1}}$$
$$=\frac{\mue}{\ndraw_j - \dimval -1}\text{tr}\br{I_{\dimval}}$$
$$=\frac{\mue}{\ndraw_j - \dimval -1}\cd \dimval$$
Next, we can reduce the linear regression case to mean estimation. In this case, assume a 1-dimensional input with $x=1$ deterministically. After drawing $x_j$, we multiply it by $\param_j$ and add some noise governed by $\err_j$: this is the exact same structure as mean estimation. In this case, $\cross{j} = \expXlin{j}[\X_j^2] + Var(\X_j) = 1 + 0= 1$. Similarly, $\X_j^T\X_j = \ndraw_j$ deterministically, so the error term reduces to $\frac{\mue}{\ndraw_j}$ as desired. 

Note that, as expected, this does \emph{not} simplify down to the mean estimation case for $\dimval=1$: that case would model a version of 1-dimensional linear regression, where it is necessary to estimate both $\param$ as well as $\hat x$, the mean of the input distribution.
\end{proof}

Next, we calculate expected MSE for the fine-grained linear regression case. 

\linfedv*

\begin{proof}
Here, we will use $\expdata{i}$ to mean the expectation taken over data from all players $i \in [\nplayer]$, given that all of the data influences the federated learning result. 
$$(\x^T\param_j - \x^T\hat{\param}^{\vm}_j)^2 $$
$$=(\x^T\param_j -\x^T\param^{\vm}_j + \x^T\param^{\vm}_j- \x^T\hat{\param}^{\vm}_j)^2 $$
\begin{equation}\label{twoterms}
\begin{split}
=(\x^T\param_j - \x^T\param^{\vm}_j)^2 + (\x^T\param^{\vm}_j - \x^T\hat{\param}^{\vm}_j)^2 \\
+ 2(\x^T\param_j - \x^T\param^{\vm}_j)\cd (\x^T\param^{\vm}_j - \x^T\hat{\param}^{\vm}_j)
\end{split}
\end{equation}
Note that the expectation of the last term in Equation \ref{twoterms} results in 0 because $\expdata{j}\br{\x^T\param^{\vm}_j - \x^T\hat{\param}^{\vm}_j}=0$. Next, we investigate the second equation in Equation \ref{twoterms}. 
$$(\x^T\param^{\vm}_j - \x^T\hat{\param^{\vm}_j})^2$$
$$= \p{\x^T\sum_{i=1}^{\nplayer}\vm_{ji}\param_i - \x^T\sum_{i=1}^{\nplayer}\vm_{ji}\hat{\param}_i}^2$$
$$= \p{\sum_{i=1}^{\nplayer}\vm_{ji}\x^T(\param_i -\hat{\param}_i)}^2$$
Expanding out the squared term gives us: 
$$\sum_{i=1}^{\nplayer}\p{\vm_{ji}\x^T(\param_i-\hat{\param}_i)}^2 $$
$$+ \sum_{i=1}^{\nplayer}\sum_{k \ne i}\p{\vm_{ji} \cd \x^T(\param_i - \hat{\param}_i) \cd \vm_{jk} \cd \x^T(\param_k - \hat{\param}_k)}$$
The second term ends up being irrelevant: because each set of parameters $\param_i \sim \gendist$ are drawn independently and because each data set $\X_i \sim \xdist{i}$ are drawn independently, the $\param_i - \hat{\param}_i$ terms are independent of each other. Because each is 0 in expectation, the entire product has expectation 0. Rewriting the first term gives: 
$$\sum_{i=1}^{\nplayer}\vm_{ji}^2\cd (\x^T\param_i - \x^T\hat{\param}_i)^2$$
The term inside the sum is exactly equivalent to the value we solved with the local estimation case: we can rewrite this as
$$\mue\sum_{i=1}^{\nplayer}\vm_{ji}^2\cd \text{tr}[\cross{j}\expdata{i}\br{(\X_i^T\X_i)^{-1}}$$
or, if the necessary conditions are satisfied, 
$$\mue\sum_{i=1}^{\nplayer}\vm_{ji}^2\cd \frac{\dimval}{\ndraw_i - \dimval -1}$$
Finally, we will explore the first term Equation \ref{twoterms}:
$$(\x^T\param_j - \x^T\param^{\vm}_j)^2$$
$$=(\x^T(\param_j-\param^{\vm}_j))^T\x^T(\param_j - \param^{\vm}_j)$$
$$=(\param_j-\param^{\vm}_j)^T\x\x^T(\param_j - \param^{\vm}_j)$$
Taking the expectation and using the cyclic property of the trace gives: 
$$=\text{tr}\br{(\param_j-\param^{\vm}_j)^T\expxlin{j}[\x\x^T](\param_j - \param^{\vm}_j)}$$
\begin{equation}\label{trace}
 =\text{tr}\br{\cross{j}(\param_j - \param^{\vm}_j)(\param_j-\param^{\vm}_j)^T}  
\end{equation}
Next, we focus on simplifying the inner term of Equation \ref{trace} involving the $\param$ values. Using the definition of $\param^{\vm}_j$ gives: 
$$\p{\param_j - \sum_{i=1}^{\nplayer}\vm_{ji}\param_i}\p{\param_j - \sum_{i=1}^{\nplayer}\vm_{ji}\param_i}^T$$
$$=\p{(1-\vm_{jj}) \cd \param_j - \sum_{i \ne j}\vm_{ji}\cd \param_i}$$
$$\cd \p{(1-\vm_{jj}) \cd \param_j - \sum_{i \ne j}\vm_{ji}\cd \param_i}^T$$
We know that $1 = \vm_{jj} + \sum_{i\ne j}\vm_{ji}$, so we can rewrite this as:
$$=\p{\sum_{i\ne j} \vm_{ji}\cd \param_j - \sum_{i \ne j}\vm_{ji}\cd \param_i}\p{\sum_{i\ne j}\vm_{ji}\cd \param_j - \sum_{i \ne j}\vm_{ji}\cd \param_i}^T$$
$$=\p{\sum_{i\ne j} \vm_{ji}\cd (\param_j - \param_i)}\p{\sum_{i\ne j}\vm_{ji}\cd (\param_j - \param_i)}^T$$
Expanding gives us two terms: 
\begin{equation}\label{twotermsb}
\begin{split}
\sum_{i\ne j}&\vm_{ji}^2(\param_j - \param_i)(\param_j - \param_i)^T\\
+ \sum_{i, k \ne j, i\ne k}&\vm_{ji} \cd \vm_{jk} \cd (\param_j - \param_i) (\param_j - \param_k)^T
\end{split}
\end{equation}
Note that $$(\param_j - \param_i)\cd (\param_j - \param_k) = \param_j\param_j^T -\param_j\param_k^T -\param_i\param_j^T+\param_i\param_k^T$$
Because the parameters are drawn iid, 
$$\mathbb{E}_{\param_j \sim \gendist}\br{\param_j \param_j^T} = \mathbb{E}_{\param_j \sim \gendist}\br{\param_i \param_i^T} \ \forall i \in [\nplayer]$$
$$\mathbb{E}_{\param_j \sim \gendist}\br{\param_j \param_k^T} = \mathbb{E}_{\param_j \sim \gendist}\br{\param_i \param_l^T} \ \forall j\ne k, i \ne l, \ i, j, k, l \in [\nplayer]$$
$\mathbb{E}_{\param_j \sim \gendist}[\param_j\param_j^T]$ is implied to mean the same thing as $\mathbb{E}_{\param_j \sim \gendist}[\param_i\param_i^T]$. Using these results allows us to expand out the product and take the expectation. The first term in Equation \ref{twotermsb} becomes: 
$$\sum_{i\ne j}\vm_{ji}^2\cd \p{\mathbb{E}_{\param_j \sim \gendist}\br{\param_j \param_j^T} + \mathbb{E}_{\param_j \sim \gendist}\br{\param_i \param_i^T} - 2\mathbb{E}_{\param_j \sim \gendist}\br{\param_i \param_j^T}}$$
$$=2\sum_{i\ne j}\vm_{ji}^2\cd \p{\mathbb{E}_{\param_j \sim \gendist}\br{\param_j \param_j^T} - \mathbb{E}_{\param_j \sim \gendist}\br{\param_i \param_j^T}}$$
For the second term in Equation \ref{twotermsb}, each product within the sum becomes : 
$$\mathbb{E}_{\param_j \sim \gendist}\br{\param_j \param_j^T} -\mathbb{E}_{\param_j \sim \gendist}\br{\param_j \param_k^T}-\mathbb{E}_{\param_j \sim \gendist}\br{\param_i \param_j^T} +\mathbb{E}_{\param_j \sim \gendist}\br{\param_i \param_k^T}$$
The last two terms cancel because the parameters are drawn i.i.d., so the sum becomes: 
$$=\mathbb{E}_{\param_j \sim \gendist}\br{\param_j \param_j^T} -\mathbb{E}_{\param_j \sim \gendist}\br{\param_j \param_k^T}$$
The overall sum in Equation \ref{twotermsb} becomes: 
$$=\sum_{i, k \ne j, i\ne k}\vm_{ji} \cd \vm_{jk} \cd \p{\mathbb{E}_{\param_j \sim \gendist}\br{\param_j \param_j^T} -\mathbb{E}_{\param_j \sim \gendist}\br{\param_j \param_k^T}}$$
The sum of both of these terms taken together becomes: 
$$\p{\mathbb{E}_{\param_j \sim \gendist}\br{\param_j \param_j^T} -\mathbb{E}_{\param_j \sim \gendist}\br{\param_j \param_k^T}} \cd \p{2\sum_{i\ne j}\vm_{ji}^2 + \sum_{i, k \ne j, i\ne k}\vm_{ji}}$$
Next, we can use the identity: 
$$\sum_{i, k\ne j, i\ne k}\vm_{ji} \cd \vm_{jk} = \p{\sum_{i\ne j}\vm_{ji}}^2 - \sum_{i\ne j}\vm_{ji}^2$$
Which allows us to rewrite as: 
$$\p{\mathbb{E}_{\param_j \sim \gendist}\br{\param_j \param_j^T} -\mathbb{E}_{\param_j \sim \gendist}\br{\param_j \param_k^T}} \cd \p{\sum_{i\ne j}\vm_{ji}^2 + \p{\sum_{i\ne j}\vm_{ji}}^2}$$
We can alternatively write this a different way. Because $1 = \vm_{jj} + \sum_{i\ne j}\vm_{ji}$, we have $1 - \vm_{jj}= \sum_{i\ne j}\vm_{ji}$
$$\p{\mathbb{E}_{\param_j \sim \gendist}\br{\param_j \param_j^T} -\mathbb{E}_{\param_j \sim \gendist}\br{\param_j \param_k^T}} \cd \p{\sum_{i\ne j}\vm_{ji}^2 + \p{1-\vm_{jj}}^2}$$
Recall that this analysis was focusing solely on the component of Equation \ref{trace} that involved the $\param$ product. We can recombine our simplification into Equation \ref{trace} to rewrite it as: 
$$\text{tr}\br{\cross{j} \p{\mathbb{E}_{\param_j \sim \gendist}\br{\param_j \param_j^T} - \mathbb{E}_{\param_i, \param_j \sim \gendist}\br{\param_j\param_i^T}}}$$
$$\cd \p{\sum_{i\ne j}\vm_{ji}^2 + \p{1-\vm_{jj}}^2}$$
Next, we need to reason about the difference in the expected terms. In this setting, we are assuming that each coefficient is drawn separately from the other coefficients. $\mathbb{E}_{\param_j \sim \gendist}\br{\param_j \param_j^T}$ has, on the $d$th element of the diagonal, $\mathbb{E}_{\param_j \sim \gendist}[(\param_j^d)^2]$ and on the off-diagonal terms in the $l, k$th entry, has $\mathbb{E}_{\param_j \sim \gendist}[\param_j^l \cd \param_j^k]$ Here, we are assuming that $\param_j^l$ and $\param_j^k$ are independent, so this equals $\mathbb{E}_{\param_j \sim \gendist}[\param_j^l] \cd \mathbb{E}_{\param_j \sim \gendist}[\param_j^k]$: we relax this assumption below.  $\mathbb{E}_{\param_i, \param_j \sim \gendist}\br{\param_j \param_i^T}$ has, on the $d$th element of the diagonal, $\mathbb{E}_{\param_j \sim \gendist}[(\param_j^d)]^2$, and on the $l, k$th off-diagonal term has the same value as the other matrix: $\mathbb{E}_{\param_j \sim \gendist}[\param_j^l] \cd \mathbb{E}_{\param_j \sim \gendist}[\param_j^k]$. The difference between these two matrices is a diagonal matrix with $\mathbb{E}_{\param_j \sim \gendist}[(\param_j^d)^2] -\mathbb{E}_{\param_j \sim \gendist}[(\param_j^d)]^2 = \var_d$ on the diagonal, where $\var_d$ represents the variance of the $d$th coefficient. That turns the term involving the trace into a simple sum: 
$$\p{\sum_{i\ne j}\vm_{ji}^2 + \p{\sum_{i\ne j}\vm_{ji}}^2}\cd \sum_{d=1}^{\dimval}\expxlin{j}[(\x^d)^2] \cd \var_d$$

In the proof above, we assumed that the draw of parameter value $\param_j^l$ is independent of $\param_j^k$, for $l \ne k$. A case where this might not occur is when these values are correlated: say, the value drawn for $\param_j^l$ is anti-correlated with the parameter drawn for $\param_j^k$. (Note that we still assume draws are independent across players: $\param_j^l$ is independent of $\param_i^l$ and $\param_i^k$). Relaxing this assumption is not hard and would change the results in the following way: the off-diagonal terms of the difference would no longer be 0. Instead, the off-diagonal $l, k$th entry becomes
$$\mathbb{E}_{\param_j \sim \gendist}[\param_j^l \cd \param_j^k]-\mathbb{E}_{\param_j \sim \gendist}[\param_j^l] \cd \mathbb{E}_{\param_j \sim \gendist}[\param_j^k]$$
Performing the matrix multiplication with $\cross{j}$ turns this into: 
$$\sum_{d=1}^{\dimval}\expxlin{j}[(\x^d)^2] \cd \var_d + \sum_{l \ne d}\expxlin{j}[\x^d \cd \x^l]$$
$$\cd (\mathbb{E}_{\param_j \sim \gendist}[\param_j^l \cd \param_j^k]-\mathbb{E}_{\param_j \sim \gendist}[\param_j^l] \cd \mathbb{E}_{\param_j \sim \gendist}[\param_j^k])$$
Our final value for this component of the error would be the same form, but with a slightly different coefficient.

Finally, we will consider the mean estimation case. As before, we note that $\cross{j} = 1$ and $\X_j^T\X_j = \ndraw_j$ deterministically, so that component of the error term reduces to 
$$\mue \sum_{i=1}^{\nplayer}\vm_{ji}^2 \cd \frac{1}{\ndraw_j}$$
similarly, we note that $\expxlin{j}[(\x^d)^2] = 1$, so the second component reduces to: 
$$\p{\sum_{i\ne j}\vm_{ji}^2 + \p{\sum_{i\ne j}\vm_{ji}}^2}\var$$
\end{proof}

\linfed*

\begin{proof}
To obtain this result, we note that the uniform federation case amounts to weights $\vm_{ji} = \frac{\ndraw_i}{\total}$. For both linear regression and mean estimation, the $\var$ multiplier becomes:
$$\sum_{i\ne j}\vm_{ji}^2 + \p{\sum_{i\ne j} \vm_{ji}}^2$$
$$ = \frac{1}{\total^2}\p{\sum_{i\ne j} \ndraw_i^2 + \p{\sum_{i\ne j} \ndraw_i}^2}$$
For linear regression, the $\mue$ multiplier becomes:
$$\sum_{i=1}^{\nplayer}\vm_{ji}^2 \cd \frac{\dimval}{\ndraw_i - \dimval -1}$$
$$=\frac{1}{\total^2}\sum_{i=1}^{\nplayer}\ndraw_i^2 \cd \frac{\dimval}{\ndraw_i - \dimval -1}$$
And for mean estimation, the $\mue$ multiplier becomes: 
$$\sum_{i=1}^{\nplayer}\vm_{ji}^2 \cd \frac{1}{\ndraw_i}$$
$$= \sum_{i=1}^{\nplayer}\frac{\ndraw_i^2}{\total^2} \cd \frac{1}{\ndraw_i} = \sum_{i=1}^{\nplayer}\frac{\ndraw_i}{\total^2} = \frac{1}{\total}$$
\end{proof}

\linfedw*

\begin{proof}
To obtain this result, we note that the coarse-grained case corresponds to $\vm_{jj} = \w + \frac{(1-\w) \cd \ndraw_j}{\total}$ and $\vm_{ji} = (1-\w) \cd \frac{\ndraw_i}{\total}$. For both linear regression and mean estimation, the $\var$ multiplier becomes: 
$$\sum_{i \ne j}\vm_{ji}^2 + \p{\sum_{i\ne j} \vm_{ji}}^2$$
$$ =(1-\w)^2\cd \frac{\sum_{i\ne j} \ndraw_i^2 + \p{\sum_{i\ne j}\ndraw_i}^2}{\total^2}$$
For linear regression, let $\zeta_i$ stand for either $\text{tr}[\cross{j}\expdata{i}\br{(\X_i^T\X_i)^{-1}}$ or  $\frac{\dimval}{\ndraw_i - \dimval -1}$, depending on the linear regression case. Then, the $\mue$ multiplier becomes: 
$$\sum_{i=1}^{\nplayer}\vm_{ji}^2 \cd \zeta_i$$
$$ = \p{\w + \frac{(1-\w) \cd \ndraw_j}{\total}}^2 \cd \zeta_j + \sum_{i\ne j} (1-\w)^2 \cd \frac{\ndraw_i^2}{\total^2} \cd \zeta_j$$
$$ = \zeta_j  \cd \p{\w^2 + 2 \frac{(1-\w) \cd \w \cd \ndraw_j}{\total}} +$$
$$+ (1-\w)^2 \sum_{i=1}^{\nplayer}\frac{\ndraw_i^2}{\total^2}\cd \zeta_j$$
For mean estimation, the $\mue$ multiplier becomes: 
$$\sum_{i=1}^{\nplayer} \vm_{ji}^2 \cd \frac{1}{\ndraw_i}$$
$$ = \p{\w + \frac{(1-\w) \cd \ndraw_j}{\total}}^2 \cd \frac{1}{\ndraw_j} + \sum_{i\ne j} (1-\w)^2 \cd \frac{\ndraw_i^2}{\total^2} \cd \frac{1}{\ndraw_i}$$
$$= \frac{\w^2}{\ndraw_j} + 2 \frac{(1-\w) \cd \w}{\total} + \frac{(1-\w)^2 \cd \ndraw_j}{\total^2} + \sum_{i\ne j} (1-\w)^2 \cd \frac{\ndraw_i}{\total^2} $$
$$ = \frac{\w^2}{\ndraw_j} + 2 \frac{(1-\w) \cd \w}{\total} + (1-\w)^2 \sum_{i=1}^{\nplayer}\frac{\ndraw_i}{\total^2}$$
$$ = \frac{\w^2}{\ndraw_j} + 2 \frac{(1-\w) \cd \w}{\total} + \frac{(1-\w)^2}{\total}$$
$$ = \frac{\w^2}{\ndraw_j}+ \frac{1-\w^2}{\total}$$
\end{proof}

\section{Supporting proofs for uniform federation}\label{app:coalition}

\stabsame*
\begin{proof}
If a partition $\pi$ is optimal for every player, then it is core stable: there does not exist a coalition $C$ where all players prefer $C$ to $\pi$, because there does not exist a coalition where \emph{any} players prefer $C$ to $\pi$. 

If a partition $\pi$ is optimal for every player, then no other partition can be core stable: any set of players not in their optimal configuration could form a coalition $C$ where all players would prefer $C$. 

In the case that players are indifferent between any arrangement, then for any partition $\pi$ and any competing coalition $C$, all players would be indifferent between $\pi$ and $C$, so $\pi$ is core stable. 
\end{proof}

\allbigeq*
\begin{proof}
First, we will consider the case where the inequality is strict and will show the statement is satisfied by showing that every player minimizes their error by being alone in $\alone$. We will use $\total_Q$ to be the sum of elements within a coalition $Q$. $Q$ could be the coalition equal to all players (\gcol) or some strict subset, but we will assume it contains at least 2 elements. We will show that every player gets higher error in $Q$ than it would get alone. We wish to show: 
$$\frac{\mue}{\total_Q} + \var \frac{\sum_{i\ne j, i \in Q}\ndraw_i^2 + (\total_Q - \ndraw_j)^2}{\total_Q^2} > \frac{\mue}{\ndraw_j}$$
Cross multiplying gives: 
$$\mue \cd \total_Q \cd \ndraw_j + \var \cd \ndraw_j\cd \p{\sum_{i\ne j, i \in Q}\ndraw_i^2 + (\total_Q - \ndraw_j)^2} > \mue \cd \total_Q^2 $$
Rewriting: 
$$\var \cd \ndraw_j\cd \p{\sum_{i\ne j, i \in Q}\ndraw_i^2 + (\total_Q - \ndraw_j)^2} > \mue \cd \total_Q^2 -\mue \cd \total_Q \cd \ndraw_j $$
The righthand side can be rewritten as:
$$\mue \cd \total_Q^2 -\mue \cd \total_Q \cd \ndraw_j= \mue \cd(\total_Q-\ndraw_j)^2 + \mue \cd \ndraw_j \cd (\total_Q - \ndraw_j) $$
Then, we can prove the inequality by splitting it up into two terms. The first: 
$$\var \cd \ndraw_j \cd (\total_Q - \ndraw_j)^2 > \mue \cd (\total_Q -\ndraw_j)^2$$
which is true because $\ndraw_j \cd \var > \mue$. The second: 
$$\var \cd \ndraw_j \cd \sum_{i\ne j, i \in Q}\ndraw_i^2 > \mue \cd \ndraw_j \cd (\total_Q - \ndraw_j)$$
$$\var \cd\sum_{i\ne j, i \in Q}\ndraw_i^2 > \mue \cd \sum_{i \ne j, i \in Q}\ndraw_i$$
which is satisfied because, for each player, 
$$\var \cd \ndraw_i^2 > \mue \cd \ndraw_i$$
because $\var \cd \ndraw_i > \mue$. \newline 
Next, we will consider the case where the inequality may not be strict. We can note that, in the proof above, any coalition $Q$ with at least one player $\ndraw_i > \frac{\mue}{\var}$ would satisfy the desired inequality: all players participating would get higher error than they could alone. This shows that any coalition involving a player with more samples than $\frac{\mue}{\var}$ is infeasible. We have previously shown that all players with $\ndraw_i = \frac{\mue}{\var}$ get equal error no matter their arrangement. 
\end{proof}
As a reminder, we will use the notation $\pi(\s, \el)$ to mean a coalition with $\s$ small players and $\el$ large players. The next few lemmas describe how the errors of large and small players in a coalition change as $\s$ and $\el$ are increased. 

\begin{restatable}{lemma}{sprefs}
\label{sprefs}
For uniform federation, if $\ns \leq \frac{\mue}{\var}$ and $\ns < \nl$, \textbf{small} players always see their error decrease with the addition of more \textbf{small} players: 
$$\s_2 > \s_1 \quad \Rightarrow \quad \pi(\s_2, \el) \cg_{\si} \pi(\s_1, \el)$$
\end{restatable}
\begin{proof}
We will show that the derivative of the small player's error with respect to $\s$ is always negative. The error is: 
$$\frac{\mue}{\s \cd \ns + \el\cd \nl} $$
$$+ \var \frac{(\s-1) \cd \ns^2 + \el\cd \nl^2 + ((\s-1) \cd \ns + \el\cd nl)^2}{(\s \cd \ns + \el\cd \nl)^2}$$
The derivative with respect to $\s$ is: 
$$\frac{\ns \cd (\s \cd \ns \cd (\ns \cd \var - \mue))}{(\s \cd \ns + \el\cd n_{\ell})^3}$$

$$-\frac{\ns \cd (\el\cd \nl \cd (\mue + 2 \nl \cd \var - 3 \ns \cd \var))}{(\s \cd \ns + \el\cd n_{\ell})^3}$$
Showing that the derivative is negative is equivalent to showing that the term below is negative: 
$$\s \cd \ns \cd (\ns \cd \var - \mue)-\el\cd \nl \cd (\mue + 2 \nl \cd \var - 3 \ns \cd \var)$$
We can break this term into multiple components: 
$$\s \cd \ns \cd (\ns \cd \var - \mue) \leq  0$$
because $\ns \leq \frac{\mue}{\var}$. We can rewrite a second term as: 
$$\mue - \ns \cd \var + 2 \var(\nl - \ns)$$
We know that $\mue - \ns \cd \var \geq 0$, 
and because $\nl > \ns$, 
$$2 \nl \cd \var - 2 \ns \cd \var >0$$
These facts, taken together, show that the derivative is always negative.
\end{proof}

\begin{restatable}{lemma}{lprefl}
\label{lprefl}
For uniform federation, if $\nl \geq \frac{\mue}{\var}$ and $\ns < \nl$, \textbf{large} players see their error increase with the addition of more \textbf{large} players, so they prefer $\el$ as small as possible:
$$\el_2 < \el_1 \quad \Rightarrow \quad \pi(\s, \el_2) \cg_{\li} \pi(\s, \el_1)$$
If $\nl < \frac{\mue}{\var}$, then there exists some
$\el_0$ such that for all $\el' < \el_0$, the derivative of the large players' error with respect to $\el$ is positive, and for all $\el'>\el_0$, the derivative of their error is negative. 
\end{restatable}
\begin{proof}
To prove this, we will show that the derivative of the large player's error with respect to $\el$ is always positive when $\nl \geq \frac{\mue}{\var}$. The error is: 
$$\frac{\mue}{\s \cd \ns + \el\cd \nl} $$
$$+ \var \frac{\s \cd \ns^2 + (\el-1) \cd \nl^2 + (\s \cd \ns + (\el-1) \cd \nl)^2}{(\s \cd \ns + \el\cd \nl)^2}$$
The derivative with respect to $\el$ is: 
$$\frac{\nl (\el\cd \nl (\nl \var - \mue) - \ns \cd s (\mue - 3 \nl \cd \var + 2 \ns \cd \var))}{(\el\cd \nl + \ns \cd \s)^3}$$
We wish to show that the numerator is positive. We can break it into multiple components: 
$$\nl \cd \var - \mue \geq 0$$
because $\nl \geq \frac{\mue}{\var}$. We can rewrite the second term as
$$\mue - \nl \cd \var + 2 \var(\ns - \nl)$$
which is negative because $\nl \geq \frac{\mue}{\var}$ and $\nl >\ns$.\newline 
Next, we consider the case where $\nl < \frac{\mue}{\var}$. The first term is now negative and the second term is the sum of two terms: one is positive and one is negative. 
The overall derivative is negative whenever: 
$$\el\cd \nl (\nl \var - \mue) - \ns \cd s (\mue - 3 \nl \cd \var + 2 \ns \cd \var)< 0$$
$$\el\cd \nl (\nl \var - \mue) < \ns \cd s (\mue - 3 \nl \cd \var + 2 \ns \cd \var)$$
$$\el > \frac{\ns \cd s (\mue - 3 \nl \cd \var + 2 \ns \cd \var)}{\nl \cd (\nl \cd \var - \mue)}$$
Note that, for the assumption, the denominator is negative. If the numerator is positive, then this is true for all $\el > 0$, so the slope is always negative. If the numerator is negative, then the slope is negative for all $\el > \el_0$ for the given $\el_0 > 0$. 
\end{proof}

\begin{restatable}{lemma}{sprefl}
\label{sprefl}
Assume uniform federation with $\ns \leq \frac{\mue}{\var}$ and $\ns < \nl$. If $\nl \leq  \frac{\mue}{\var}$, \textbf{small} players always prefer federating with more \textbf{large} players: 
$$\el_2 > \el_1 \quad \Rightarrow \quad \pi(\s, \el_2) \cg_{\si} \pi(\s, \el_1)$$
If $\nl > \frac{\mue}{\var}$, then there exists some $\el_1$ such that for all $\el' < \el_1$, the derivative of the small players' error with respect to $\el$ is positive, and for all $\el'>\el_1$, the derivative of their error is negative.  
\end{restatable}
\begin{proof}
The small player's error is
$$\frac{\mue}{\s \cd \ns + \el\cd \nl} $$
$$+ \var \frac{(\s-1) \cd \ns^2 + \el\cd \nl^2 + ((\s-1) \cd \ns + \el\cd nl)^2}{(\s \cd \ns + \el\cd \nl)^2}$$
The derivative with respect to $\el$ is: 

$$-\frac{\nl \cd (\el\cd \nl \cd (\mue -\var \cd \ns + \var (\nl - \ns))}{(\el\cd \nl + \s \cd \ns)^3}$$

$$-\frac{\nl \cd (\s \cd \ns \cd (\mue - \nl \cd \var)}{(\el\cd \nl + \s \cd \ns)^3}$$
The derivative is negative when the term below is positive: 
$$\el\cd \nl \cd (\mue -\var \cd \ns + \var (\nl - \ns)) + \s \cd \ns \cd (\mue - \nl \cd \var)$$
The first term (multiplying $\el\cd \nl$) is always positive. For $\nl \leq \frac{\mue}{\var}$ the second term is also positive or zero, so the derivative is always negative. 

If $\nl > \frac{\mue}{\var}$, then second term (multiplying $\s$) is negative. The overall derivative is negative when the term below is positive: 
$$\el\cd \nl \cd (\mue -\var \cd \ns + \var (\nl - \ns)) + \s \cd \ns \cd (\mue - \nl \cd \var)>0$$
$$\el\cd \nl \cd (\mue -\var \cd \ns + \var (\nl - \ns)) > \s \cd \ns \cd (\nl \cd \var - \mue)$$
$$\el > \frac{\s \cd \ns \cd (\nl \cd \var - \mue)}{\nl \cd (\mue -\var \cd \ns + \var (\nl - \ns))}$$
\end{proof}

\begin{restatable}{lemma}{lprefs}
\label{lprefs}
Assume uniform federation with $\ns\leq \frac{\mue}{\var}$ and $\ns< \nl$. If $\nl \geq \frac{\mue}{\var}$, then there exists some $\s_0$ such that for all $\s' < \s_0$, the derivative of the \textbf{large} players' error with respect to the number of \textbf{small} players $\s$ is negative, and for all $\s > \s_0$ the derivative of their error is positive. \newline 
If $\nl < \frac{\mue}{\var}$, then the shape of the curve either first increases, then decreases for all $\s' > \s_0$, or else first decreases, and then increases for all $\s' > \s_0$.  
\end{restatable}
\begin{proof}
The large player's error is: 
$$\frac{\mue}{\s \cd \ns + \el\cd \nl}$$
$$+ \var \frac{\s \cd \ns^2 + (\el-1) \cd \nl^2 + (\s \cd \ns + (\el-1) \cd nl)^2}{(\s \cd \ns + \el\cd \nl)^2}$$
The derivative with respect to $\s$ is: 
$$- \frac{\ns \cd (\el\cd \nl (\mue - \ns \cd \var))}{(\el \cd \nl + \s \cd \ns)^3}$$
$$- \frac{\ns \cd (\ns \cd \s \cd (\mue - \nl \cd \var + \var(\ns - \nl)))}{(\el \cd \nl + \s \cd \ns)^3}$$
The derivative is negative when the term below is positive: 
$$\el \cd \nl (\mue - \ns \cd \var) + \ns \cd \s \cd (\mue - \nl \cd \var + \var(\ns - \nl)) $$
For $\nl \geq \frac{\mue}{\var}$, the first term is always positive or zero and the second term is always negative. 
Solving for when the overall derivative is positive gives:  
$$\el \cd \nl (\mue - \ns \cd \var) + \ns \cd \s \cd (\mue - \nl \cd \var + \var(\ns - \nl))>0$$
$$ \ns \cd \s \cd (\mue - \nl \cd \var + \var(\ns - \nl))>-\el \cd \nl (\mue - \ns \cd \var)$$
$$ \s <\frac{-\el \cd \nl (\mue - \ns \cd \var)}{\ns \cd (\mue - \nl \cd \var + \var(\ns - \nl))} $$
\newline 
Next, we consider the case where $\nl < \frac{\mue}{\var}$. Again, the first term is positive or zero. The second term, though, is composed of a sum: one component is positive and one is negative, so it is not necessarily clear whether the overall sum is positive or negative. If the coefficient on $\s$ is negative, then the curve has the same shape as in the $\nl \geq \frac{\mue}{\var}$ case: first decreasing, and then increasing. If the coefficient is positive, then by the same analysis as before, the derivative is negative whenever the following inequality holds: 
$$ \ns \cd \s \cd (\mue - \nl \cd \var + \var(\ns - \nl))>-\el \cd \nl (\mue - \ns \cd \var)$$
$$ \s >\frac{-\el \cd \nl (\mue - \ns \cd \var)}{\ns \cd (\mue - \nl \cd \var + \var(\ns - \nl))} $$
In this case, the derivative of the large players' error first increases, then decreases. The shape of the curve depends on more information, which indicates whether the coefficient on $\s$ is positive or negative. 
\end{proof}

\smallcore*
\begin{proof}
In this proof, we will use the results of Lemmas \ref{sprefs}, \ref{lprefl}, \ref{sprefl}, and \ref{lprefs}. 

The small players always prefer $\s$ as large as possible, and for $\nl \leq \frac{\mue}{\var}$ they also prefer $\el$ as large as possible, so $\pi(\Sv, \Lv) = \gcol$ minimizes error for small players. For this reason, any defection coalition that has $\pi(\s>0, \el)$ is infeasible because the small players would get higher error. 

The only kind of defections we need to consider are in the form $\pi(0, \el)$. We will consider $\pi(0, \Lv)$ and show that the large players prefer $\pi(\Sv, \Lv)$ to $\pi(0, \Lv)$: $\pi(\Sv, \Lv)\cg_{\li}\pi(0, \Lv)$. In the case that $\nl < \frac{\mue}{\var}$, $\pi(0,\Lv) \cg_{\li} \pi(0, \el < \Lv)$, so any other arrangement is also not a possible defection. In the case that $\nl = \frac{\mue}{\var}$, $\pi(0,\Lv) =_{\li} \pi(0, \el < \Lv)$, so similarly any other defection is not possible. 
What we'd like to show is: 
$$\frac{\mue}{\s \cd \ns + \el \cd \nl} $$
$$+ \var \frac{\s \cd \ns^2 + (\el -1) \cd \nl^2 + (\s \cd \ns + (\el -1) \cd \nl)^2}{(\s \cd \ns + \el \cd \nl)^2} $$
$$< \frac{\mue}{\el \cd \nl} + \var\frac{(\el -1) \cd \nl^2 + (\el -1)^2 \cd \nl^2}{\el^2 \cd \nl^2}$$
Cross multiplying turns the condition into: 
$$\mue \cd (\s \cd \ns + \el \cd \nl) \cd \el^2 \cd \nl^2 $$
$$+ \var \cd (\s\cd \ns^2 + (\el -1) \cd \nl^2 + (\s \cd \ns + (\el -1) \cd \nl)^2) \cd \el^2 \cd \nl^2$$
$$< \mue \cd \el \cd \nl \cd (\s \cd \ns + \el \cd \nl)^2 + \var\cd \nl^2 \cd (\el -1) \cd \el \cd (\s \cd \ns + \el \cd \nl)^2$$
If we collect the $\mue$ terms, we get:
$$\mue \cd \el \cd \nl \cd (\s \cd \ns + \el \cd \nl) \cd (\s \cd \ns + \el \cd \nl - \el \cd \nl)$$
$$=\mue \cd \el \cd \nl \cd (\s \cd \ns + \el \cd \nl) \cd \s \cd \ns$$
If we collect the $\var$ terms, we get: 
$$\var \cd \el\cd \nl^2 \cd (\el  \cd (\s\cd \ns^2 + (\el -1) \cd \nl^2  $$
$$+ (\s \cd \ns + (\el -1) \cd \nl)^2) - (\el -1) \cd (\s \cd \ns + \el \cd \nl)^2)$$
First, we expand the first squared term and combine it with another term: 
$$\s \cd \ns^2 + (\el -1) \cd \nl^2 + \s^2 \cd \ns^2 + 2 \cd \s \cd (\el-1)\cd \nl + (\el -1)^2 \cd \nl^2$$
$$=\ns^2 \cd \s \cd (\s+1) + 2 \cd \s \cd (\el-1) \cd \ns\cd \nl + \nl^2 \cd (\el -1) \cd \el $$
Multiplied by $\el$, it becomes: 
$$\el \cd (\ns^2 \cd \s \cd (\s+1) + 2 \cd \s \cd (\el-1) \cd \ns\cd \nl + \nl^2 \cd (\el -1) \cd \el)$$
Expanding out the second squared term gives us: 
$$\s^2 \cd \ns^2 + 2 \s \cd \el \cd \ns \cd \nl + \el^2 \cd \nl^2$$
When we multiply this by $-(\el -1)$, it becomes
$$- (\el -1) \cd (\s^2 \cd \ns^2 + 2 \s \cd \el \cd \ns \cd \nl + \el^2 \cd \nl^2)$$
Next, we combine similar terms in both sums. First, we start with coefficients of $\ns^2$
$$\el \cd \ns^2 \cd \s \cd (\s+1) - (\el -1) \cd \ns^2 \cd \s^2 $$
$$= \ns^2 \cd \s \cd (\el \cd (\s+1) - (\el -1) \cd \s)$$
$$=\ns^2 \cd \s \cd (\el \cd \s + \el - \el \cd \s +\s) $$
$$= \ns^2 \cd \s \cd (\el +\s) $$
Next, we do the next term, which involves coefficients of $\ns \cd \nl$: 
$$2 \cd \el \cd \s \cd (\el -1) \cd \ns \cd \nl - 2 \cd (\el -1) \cd \s \cd \el \cd \ns \cd \nl $$
$$= 0$$
And the last one term, with coefficients of $\nl^2$: 
$$\el^2 \cd (\el -1) \cd \nl^2 - (\el -1) \cd \el^2 \cd \nl^2 $$
$$= 0$$
If we multiply take the only nonzero term and multiply by the terms we pulled out, it becomes: 
$$\ns^2 \cd \s \cd (\el + \s) \cd \el \cd \nl^2 \cd \var$$
Next, we return this to our inequality. What we're trying to show is: 
$$\el \cd \nl^2\cd \ns^2 \cd \s \cd (\el +\s)\cd \var < \mue \cd \el \cd \nl \cd (\s \cd \ns + \el \cd \nl) \cd \s \cd \ns$$
Cancelling some terms: 
$$ \nl\cd \ns \cd (\el +\s)\cd \var < \mue  \cd (\s \cd \ns + \el \cd \nl)$$
Expanding out terms: 
$$\var \cd (\el \cd \nl \cd \ns + \s \cd \nl \cd \ns) < \mue \cd (\s \cd \ns + \el \cd \nl)$$
We can prove this by splitting up piecewise: 
$$\var \cd \ns \cd \el \cd \nl < \mue \cd \el \cd \nl$$
because $\var \cd \ns < \mue$. Similarly, 
$$\var \cd \nl \cd \s \cd \ns \leq  \mue \cd \s \cd \ns$$
because $\var \cd \nl \leq \mue$. 
\end{proof}

\indivstab*
\begin{proof}
We will prove this directly by calculating an arrangement that is individually stable, relying on the results in Lemmas \ref{sprefs}, \ref{lprefl}, \ref{sprefl}, \ref{lprefs}. 

First, group all of the small players together. Then calculate $\ell' = \max \ell$ such that $\pi(\Sv, \ell) \succeq_{\li} \pi(0, 1)$: the largest number of large players that can be in the coalition such that the large players prefer this to being alone. Check if $\pi(\Sv, \ell) \cl_{\si} \pi(S, 0)$. If this is true, make the final arrangement $\pi(\Sv, 0), \pi(0, 1) \cdot L$: by previous lemmas around how the small player's error changes with $\el$, we know that if $\pi(S, \ell) \cl_{\si} \pi(\Sv, 0)$, then $\pi(S, \ell') \cl_{\si} \pi(\Sv, 0)$ for all $\ell' < \ell$. 

We will show that this is individually stable by showing no players wish to unilaterally deviate. 
\begin{itemize}
    \item No small player wishes to go to $\pi(1,0)$: reducing the number of small players in a group from $\Sv$ to 1 monotonically increases the error the small player faces. 
    \item No small player wishes to go to $\pi(1, 1)$: it is possible to reach this state by first going from $\pi(\Sv, 0)$ to $\pi(\Sv, 1)$ (which would increase error because $1 \leq \el'$ and $\pi(\Sv, \el) \cl_{\si} \pi(\Sv, 0)$) and then from $\pi(\Sv, 1)$ to $\pi(1, 1)$ (which would increase error because reducing the number of small players increases error). 
    \item No large player can go to $\pi(\Sv, 1)$: this would increase the error of the small players. 
    \item No large player wishes to go to $\pi(0, 2)$: this would increase the error of both large players. 
\end{itemize}

Next, we will consider the case where $\pi(\Sv, \el') \cgeq_{\si} \pi(\Sv, 0)$: in this case, we will show that $\pi(\Sv, \ell)$ is individually stable. 
\begin{itemize}
    \item No small player wishes to go from $\pi(\Sv, \el')$ to $\pi(1,0)$. We can see that $\pi(1,0)$ has higher error because we know $\pi(\Sv, \el') \cgeq_{\si} \pi(\Sv, 0) \cg_{\si} \pi(1, 0)$. 
    \item No small player wishes to go to $\pi(1,1)$. We can see $\pi(1,1)$ has higher error for the small player because $\pi(\Sv, \el') \cgeq_{\si} \pi(\Sv, 1) \cg_{\si} \pi(1,1)$. The first inequality comes from the following reasoning: if $\frac{d}{d\el}err_{\si}(\pi(\Sv,\el))$ is negative at $\el = 1$, then there is a monotonically increasing path of error from $\el'$ to 1. If  $\frac{d}{d\el}err_{\si}(\pi(\Sv,\el))$ is positive at 1, then we know that $\pi(\Sv,1) \cl_{\si} \pi(\Sv, 0)$, whereas $\pi(\Sv, \el') \cgeq_{\si} \pi(\Sv, 0)$. 
    \item No large player wishes to go to $\pi(\Sv, \el'+1)$: by definition of $\el'$, it would get greater error than in $\pi(0, 1)$. 
    \item No large player wishes to go to $\pi(0, 2)$ for the same reason as above. 
\end{itemize}
By this analysis, $\pi(\Sv, \el')$ is individually stable. 
\end{proof}

Lemma \ref{slnotbothsat}, below, will be useful in Example \ref{notstab} in analyzing whether the individually stable arrangement produced by Theorem \ref{indivstab} is also core stable. 
\begin{restatable}{lemma}{slnotbothsat}
\label{slnotbothsat}
Consider uniform federation with two coalitions $\pi(\s_1, \el_1)$ and $\pi(\s_2, \el_2)$ with $\s_2 < \s_1$. Then, it is not possible to pick $\el_2$ so that $\pi(\s_1, \el_1) \cg_{\si} \pi(\s_2, \el_2)$ and $\pi(\s_1, \el_1) \cg_{\li} \pi(\s_2, \el_2)$.
\end{restatable}
\begin{proof}
To prove this, we will rely on Lemmas \ref{sprefs}, \ref{lprefl}, \ref{sprefl}, \ref{lprefs}. 
First, we consider a hypothetical $\el_2'$ defined such that 
$$\s_1 \cd \ns + \el_1 \cd \nl = \s_2 \cd \ns + \el_2' \cd \nl$$
Note that $\el_2$ may not be an integer: we're using it as a reference tool (a hypothetical allocation) so this doesn't matter. 

First, we will show that the small player gets greater error in this case. First, we rewrite it as: 
$$\el_2' = (\s_1 - \s_2) \cd \frac{\ns}{\nl} + \el_1$$
and plug it in to the equation for the error of a small player. Note that most of the values stay the same: the entire $\mue$ term and the denominator of the $\var$ term. Similarly, note that with the given substitution, 
$$((\s_2-1) \cd \ns + \el_2'\cd \nl)^2  = ((\s_1-1) \cd \ns + \el_1\cd \nl)^2$$
The only term that changes is the other portion of the numerator, which becomes:
$$\nl^2 \cd \el_1 + \ns \cd \nl \cd (\s_1-\s_2) + (\s_2-1) \cd \ns^2$$
We would like to show that it is larger than the relevant portion of the error for $\pi(\s_1, \el_1)$, which is: 
$$\el_1 \cd \nl^2 + (\s_1-1) \cd \ns^2$$
This is equivalent to showing: 
$$\ns \cd \nl \cd (\s_1-\s_2) + (\s_2-1) \cd \ns^2 > (\s_1-1) \cd \ns^2$$
We can lower bound the lefthand side using $\nl > \ns$: 
$$\ns \cd \nl \cd (\s_1-\s_2) + (\s_2-1) \cd \ns^2 > \ns^2(\s_1-\s_2) + (\s_2-1) \cd \ns^2$$
$$= \ns^2 \cd (\s_1-1)$$
which shows that the small player gets greater error in $\pi(\s_2, \el_2')$ than in $\pi(\s_1, \el_1)$. 

Next, we'll show that the large player also gets greater error. We will first note that from the definition of the errors of the small and large players, we can write: 
$$err_{\si}(\s_1, \el_1) = err_{\li}(\s_1, \el_1) + 2\var \frac{\nl - \ns}{\el_1 \cd \nl + \s_1 \cd \ns}$$
and 
$$err_{\si}(\s_2, \el_2') = err_{\li}(\s_2, \el_2') + 2\var \frac{\nl - \ns}{\el_2' \cd \nl + \s_2 \cd \ns}$$
Note that, by the definition of $\el_2'$, the additive term for each of these qualities is the same. We also have just shown that $err_{\si}(\s_2, \el_2') > err_{\si}(\s_1, \el_1)$. From these two equalities, this must imply that $ err_{\li}(\s_2, \el_2') >  err_{\li}(\s_1, \el_1)$. 

We have shown that both the large and small players get higher error in $\pi(\s_2, \el_2')$ than $\pi(\s_1, \el_1)$. Next, we will show that they have different preferences about whether they wish $\el_2'$ were larger or smaller, which means that no matter what $\el_2$ is, it will leave at least one of them with higher error. 

First, we know from previous analysis in Lemma \ref{lprefl} that the large players always wish there were fewer other large players in a coalition: the large players want $\el_2 < \el_2'$. 

Secondly, we know from Lemma \ref{sprefl} that the error the small player experiences first increases and then decreases as $\el$ increases. Is it possible to pick an $\el_2 < \el_2'$ so that the small player gets lower error there than in $\pi(\s_1, \el_1)$?

Suppose that the derivative of $err_{\si}(\s_2, \el)$ with respect to $\el$ is positive at $\el = \el_2$: then, reducing $\el$ from $\el_2'$ to $\el_2$ might reduce the small player's error. However, for every point where the small player's derivative is positive, $err_{\si}(\s_2, \el) > err_{\si}(\s_2, 0) >err_{\si}(\Sv, 0)$: the small player would not wish to move here because it would get strictly higher error than it would get in $\pi(\Sv, 0)$.

Suppose instead that the derivative of of $err_{\si}(\s_2, \el)$ with respect to $\el$ is negative or zero at $\el = \el_2$. Then, if $\el_2 < \el_2'$, reducing the number of large players in the coalition from $\el_2'$ to $\el_2$ would \emph{increase} the error of the small players. This is also not an allocation that the small players would prefer. 

Increasing $\el$, so that $\el_2 > \el_2'$ sufficiently large, would satisfy the small player, but we already showed that it would increase the error the large player experiences. As a result, it is not possible to pick an allocation that both the small and large players prefer to $\pi(\s_1, \el_1)$.
\end{proof}

\begin{example}\label{notstab}
For uniform federation, the arrangement as produced in Theorem \ref{indivstab} is not necessarily core stable. 
\end{example}

Theorem \ref{indivstab} produces an arrangement of the form $\{\pi(\Sv, \el'), \pi(0, 1) \cd (\Lv - \el')\}$, where we note that $\el'$ could equal 0. As mentioned in the end of Section \ref{sec:unadj}, this arrangement is not necessarily core stable. However, the reason why is subtle. 

Theorem \ref{indivstab} checks that the given arrangement is individually stable, so the remaining cases to check would involve multiple players moving together to a new group.  First, we will consider deviations consisting of homogeneous groups: for example, $\pi(\s, 0)$ or $\pi(0, \el)$. From Lemma \ref{sprefs} we know that small players always prefer having more small players in their coalition, so $\pi(\s, 0) \cl_{\si} \pi(\Sv, 0) \cleq_{\si} \pi(\Sv, \el')$. Similarly, $\pi(0, \el) \cl_{\li} \pi(0, 1)$ and $\pi(0, \el) \cl_{\li} \pi(\Sv, \el')$ (if $\el' > 0$), so the large players do not wish to unilaterally deviate. 

Next, we might wonder whether some of the small players and large players in the $\pi(\Sv, \el')$ coalition might wish to deviate to some $\pi(\s, \el)$ where $\s < \Sv$: assume that $\el'>0$ for now.  Lemma \ref{slnotbothsat} shows that this is not feasible: it is not possible to find a location $\pi(\s, \el')$ with $\s < \Sv$ where both the small and large players get lower error. That is, it is not possible to have both $\pi(\s, \el) \cg_{\si} \pi(\Sv, \el')$ and $\pi(\s, \el) \cg_{\li} \pi(\Sv, \el')$. 

However, it still might be possible to have a deviating coalition. Recall from Theorem \ref{indivstab} that if $\el'>0$, $\pi(\Sv, \el') \cg_{\li} \pi(0, 1)$: the large players doing local learning get strictly greater error than the large players in $\pi(\Sv, \el')$. Could it be possible to find an arrangement so that $\pi(\s, \el) \cg_{\si} \pi(\Sv, \el')$ and $\pi(\s, \el) \cg_{\si} \pi(0, 1)$? 

The answer is yes. We show this by constructively producing an arrangement of players that is individually stable as produced by Theorem \ref{indivstab}, but is not core stable because a subset of the small players in $\pi(\Sv, \el')$ and the large players doing local learning both strictly prefer some $\pi(\s, \el)$. 

For this arrangement, we set $\mue = 100, \var = 1$ (note the larger $\mue$ value). We fix $\Sv = 70$ and $\Lv \geq 7$. Then, we calculate the error players get in various arrangements\footnote{The code to produce these calculations is given in our Github repository: \url{https://github.com/kpdonahue/model_sharing_games}.}.

First, we calculate the individually stable arrangement as given in Theorem \ref{indivstab}. We claim that $\pi(70, 3)$ satisfies this arrangement. For reference, we calculate the errors the small and large players get in various arrangements.  
$$err_{\si}(\pi(70, 3)) = 1.107322 \quad err_{\si}(\pi(70, 0)) = 1.115584$$
$$err_{\li}(\pi(70, 3)) = 0.932690 \quad err_{\li}(\pi(0, 1)) = 0.943396 $$ $$err_{\li}(\pi(70, 4)) =  0.943664  $$
Note that $\pi(70, 3) \cg_{\si} \pi(70, 0)$ and $\pi(70, 3) \cg_{\li} \pi(0, 1)$: both the small and large players would rather participate. Because $\pi(70, 4) \cl_{\li} \pi(0, 1)$, we cannot add another large player to the $\pi(70, 30)$ coalition. 

Next, we will show that the alternate coalition $\pi(68, 4)$ is one where small players and large players (those that are doing local learning) would both strictly prefer. 
$$err_{\si}(\pi(68, 4)) = 1.105263 \quad err_{\li}(\pi(68, 4)) = 0.943147$$
Note that $\pi(68, 4) \cg_{\si} \pi(70, 3)$ and $\pi(68, 3) \cg_{\li} \pi(0, 1)$. 

\section{Supporting proofs for coarse-grained federation}\label{app:wsupp}

\wfedminval*
\begin{proof}
Taking the derivative of the error with respect to $\w$ produces: 
$$2\mue\p{\frac{\w}{\ndraw_j} -\frac{\w}{\total}} - 2\frac{ \sum_{i\ne j}\ndraw_i^2  +(\total - \ndraw_j)^2}{\total^2}\cd (1-\w)\cd \var$$
Setting this equal to 0 and solving for $\w$ produces $\w_j$ equal to 
$$ \frac{(\sum_{i\ne j}\ndraw_i^2 + (\total - \ndraw_j)^2) \cd \var}{\mue \cd \total^2 \cd \p{\frac{1}{\ndraw_j} - \frac{1}{\total}} + (\sum_{i\ne j}\ndraw_i^2 + (\total - \ndraw_j)^2) \cd \var}$$
Note that this value $\w_j$ depends on the player $j$ that it is in reference to. It is also always strictly between 0 and 1. 
To confirm that this is a point of minimum error rather than maximum, we can take the second derivative of the error, which gives a result that is always positive: 
$$2\mue\p{\frac{1}{\ndraw_j} -\frac{1}{\total}} + 2\frac{ \sum_{i\ne j}\ndraw_i^2  +(\total - \ndraw_j)^2}{\total^2}\cd \var$$
\end{proof}

\wbesterr*
\begin{proof}
For conciseness, we will rewrite the error and $\w_j$ weight with $B = \sum_{i\ne j} \ndraw_i^2 + (\total -\ndraw_j)^2$. The error becomes: 
$$\mue \cd \p{\frac{\w_j^2}{\ndraw_j} + \frac{1-\w_j^2}{\total}} + \frac{B}{\total^2}\cd (1-\w_j)^2\var$$
and the $\w_j$ weight becomes
$$\frac{B \cd \var}{\mue \cd \total^2 \cd \p{\frac{1}{\ndraw_j} - \frac{1}{\total}} +B \cd \var}$$
Substituting in for $\w_j$ gives the error as: 
$$\frac{\mue}{\total} + \mue \cd \p{\frac{1}{\ndraw_j} - \frac{1}{\total}} \cd \frac{B^2 \cd (\var)^2}{\p{\mue \cd \total^2 \cd \p{\frac{1}{\ndraw_j} - \frac{1}{\total}} + B \cd \var}^2}$$
$$+ \frac{B}{\total^2} \cd \var \cd \frac{\mue^2 \cd \total^4 \cd \p{\frac{1}{\ndraw_j} - \frac{1}{\total}}^2}{\p{\mue \cd \total^2 \cd \p{\frac{1}{\ndraw_j} - \frac{1}{\total}} + B \cd \var}^2}$$
Collecting and pulling out common terms: 
$$\frac{\mue}{\total} +\frac{\mue \p{\frac{1}{\ndraw_j} - \frac{1}{\total}} \cd B \cd \var}{\p{\mue \cd \total^2 \cd \p{\frac{1}{\ndraw_j} - \frac{1}{\total}} +B \cd \var}^2}$$
$$\cd \p{B \cd\var +\mue \cd \total^2 \cd \p{\frac{1}{\ndraw_j} - \frac{1}{\total}}}$$
$$=\frac{\mue}{\total} +\frac{\mue \p{\frac{1}{\ndraw_j} - \frac{1}{\total}} \cd B \cd \var}{\mue \cd \total^2 \cd \p{\frac{1}{\ndraw_j} - \frac{1}{\total}} +B \cd \var}$$
$$= \frac{\mue^2 \cd \total^2  \p{\frac{1}{\ndraw_j} - \frac{1}{\total}} + \mue \cd B \cd \var +\total \cd \mue \p{\frac{1}{\ndraw_j} - \frac{1}{\total}} B \cd \var}{\total \cd \p{\mue \cd \total^2 \cd \p{\frac{1}{\ndraw_j} - \frac{1}{\total}} +B \cd \var}}$$
We multiply the top and bottom by $\total \cd \ndraw_j$: 
$$\frac{\mue^2 \cd \total^2 \p{\total - \ndraw_j} + \mue \cd B \cd \var \cd \total \cd \ndraw_j +\total \cd \mue \p{\total - \ndraw_j} B \cd \var}{\total \cd \p{\mue \cd \total^2 \cd \p{\total - \ndraw_j} +B \cd \var \cd \total \cd \ndraw_j}}$$
$$ =\frac{\mue^2 \cd \total \cd \p{\total - \ndraw_j} + \mue \cd B \cd \var \cd \ndraw_j + \mue \cd \p{\total - \ndraw_j} \cd B \cd \var}{\mue \cd \total^2 \cd \p{\total - \ndraw_j} +B \cd \var \cd \total \cd \ndraw_j}$$
$$= \frac{\mue^2 \cd \total \cd \p{\total - \ndraw_j} + \mue \cd B \cd \var \cd \total}{\mue \cd \total^2 \cd \p{\total - \ndraw_j} +B \cd \var \cd \total \cd \ndraw_j}$$
$$= \frac{\mue \cd \p{\total - \ndraw_j} + B \cd \var}{\total \cd \p{\total - \ndraw_j} +B \cd \frac{\var}{\mue} \cd \ndraw_j}$$
which gives the desired value. 
\end{proof}

Next, we will prove a series of lemmas describing how the error of small and large players in a coalition $\pi(\s, \el)$ using optimal coarse-grained federation see their error change as $\s$ and $\el$ increase. 

\begin{lemma}
\label{sprefsw}
For optimal coarse-grained federation, \textbf{small} players always see their error decrease with the addition of more \textbf{small} players. That is, 
$$\s_2 > \s_1 \quad \Rightarrow \quad \pi(\s_2, \el) \cg_{\si} \pi(\s_1, \el)$$
\end{lemma}
\begin{proof}
The error that the small players experience is $\frac{g_{\si}(\s, \el)}{h_{\si}(\s, \el)}$
where 
$$g_{\si}(\s, \el) = \mue \cd \p{(\s -1) \cd \ns + \el \cd \nlv} $$
$$ +\p{(\s-1) \cd \ns^2  + \el \cd \nlv^2+ \p{(\s -1) \cd \ns + \el \cd \nlv}^2} \cd \var$$
$$h_{\si}(\s, \el) = ((\s -1) \cd \ns + \el \cd \nlv)^2 + \ns \cd ((\s -1) \cd \ns + \el \cd \nlv) $$
$$+ \ns \cd \frac{\var}{\mue} \cd ((\s-1) \cd \ns^2 + \el \cd \nlv^2 + ((\s -1) \cd \ns + \el \cd \nlv)^2)$$
Next, we consider the derivative of the error term with respect to $\s$, the number of small players. If this is always negative, then small players always see their error decrease with the addition of more small players. The derivative is negative when: 
$$\frac{d}{d\s}\br{g_{\si}(\s, \el)} \cd h_{\si}(\s, \el) - g_{\si}(\s, \el) \cd \frac{d}{d\s}\br{h_{\si}(\s, \el)} < 0$$
Calculating and simplifying the lefthand side of this equation gives: 
$$- \ns \cd ((\s -1) \cd \ns + \el \cd \nlv)$$
$$\cd \p{(\s - 1) \cd \ns \cd (\mue + \ns \cd \var) + \el \cd \nlv \cd (\mue + \var \cd (2 \nlv -\ns))}$$
Each individual term is positive: note that $\nlv > \ns$, so $2 \nlv - \ns >0$. The overall term is negative, meaning that the overall derivative of the error is always negative. 
\end{proof}
\begin{lemma}
\label{spreflw}
For optimal coarse-grained federation, there exists some $\el_0$ such that for all $\el > \el_0$, the \textbf{small} players always see their error decrease with the addition of more \textbf{large} players. 
\end{lemma}
\begin{proof}
The error that the small player experiences can be given by the ratio of $g_{\si}(\cd)/h_{\si}(\cd)$ as before. In this case, we are interested in the derivative with respect to $\el$, which is negative when: 
$$\frac{d}{d\el}\br{g_{\si}(\s, \el)} \cd h_{\si}(\s, \el) - g_{\si}(\s, \el) \cd \frac{d}{d\el}\br{h_{\si}(\s, \el)} < 0$$
Calculating and simplifying the lefthand side of this equation gives: 
$$- \nlv \cd ((\s -1) \cd \ns + \el \cd \nlv)$$
$$\cd \p{(\s - 1)\ns(\mue + \var (2 \ns - \nlv)) + \el \cd \nlv (\mue + \var \cd \nlv))}$$
This term is negative whenever: 
$$\el > -\frac{\mue + \var \cd (2 \ns - \nlv)}{\nlv \cd (\mue + \var \cd \nlv)}$$
So, for sufficiently large $\el$, the derivative of the small player's error is always negative. Note that if the term on the RHS is negative, then the derivative is negative for any $\el \geq0$. 
\end{proof}

\begin{lemma}
\label{lpreflw}
For optimal coarse-grained federation, there exists some $\el_1$ such that for all $\el > \el_1$, the \textbf{large} players always see their error decrease with the addition of more \textbf{large} players. If there are no small players ($\s = 0$), then large players would most prefer to all federate together in $\pi(0, \Lv)$.
\end{lemma}
\begin{proof}
The error that the large players experience is $\frac{g_{\li}(\s, \el)}{h_{\li}(\s, \el)}$
where 
$$g_{\li}(\s, \el) = \mue \cd \p{\s \cd \ns + (\el-1) \cd \nlv} $$
$$ +\p{\s \cd \ns^2  + (\el-1) \cd \nlv^2+ \p{\s \cd \ns + (\el-1) \cd \nlv}^2} \cd \var$$
$$h_{\li}(\s, \el) = (\s \cd \ns + (\el-1) \cd \nlv)^2 + \ns \cd (\s\cd \ns + (\el-1) \cd \nlv) $$
$$+ \nlv \cd \frac{\var}{\mue} \cd (\s \cd \ns^2 + (\el-1) \cd \nlv^2 + (\s \cd \ns + (\el-1) \cd \nlv)^2)$$
The derivative of the error with respect to $\el$, the number of large players, is negative whenever:
$$\frac{d}{d\el}\br{g_{\li}(\s, \el)} \cd h_{\li}(\s, \el) - g_{\li}(\s, \el) \cd \frac{d}{d\el}\br{h_{\li}(\s, \el)} < 0$$
Calculating and simplifying the lefthand side of this equation gives: 
$$-\nlv \cd (\s \cd \ns + (\el-1) \cd \nlv)$$
$$\cd ((\el -1) \cd (\mue + \ns \cd \var) + \s \cd \ns \cd (\mue + \var \cd (2 \ns - \nlv)))$$
This term is negative whenever: 
$$\el > \frac{-\s \cd \ns \cd (\mue + \var \cd (2 \ns - \nlv)) + \mue  + \ns \cd \var}{\mue + \ns \cd \var}$$
Note that if $\s = 0$, this reduces to $\el > 1$, which says that whenever large players are arranged without small players, they would most prefer to all be together. 
\end{proof}

\begin{lemma}
\label{lprefsw}
For optimal coarse-grained federation, \textbf{large} players always see their error decrease with the addition of more \textbf{small} players. Because of this, $\pi(\s, \el) \cg_{\li} \pi(0, \el)$ for all $\s > 0$: large players would always prefer federating with more small players. 
$$\s_2 > \s_1 \quad \Rightarrow \quad \pi(\s_2, \el) \cg_{\li} \pi(\s_1, \el)$$
\end{lemma}
\begin{proof}
The error that the large players experience can be given by the ratio $g_{\li}(\cd)/h_{\li}(\cd)$ as before. In this case, we are interested in the derivative with respect to $\s$, which is negative when: 
$$\frac{d}{d\s}\br{g_{\li}(\s, \el)} \cd h_{\li}(\s, \el) - g_{\li}(\s, \el) \cd \frac{d}{d\s}\br{h_{\li}(\s, \el)} < 0$$
Gathering and simplifying the lefthand side of this equation gives: 
$$-\ns \cd (\s \cd \ns + (\el -1) \cd \nlv)$$
$$ \cd (\s \cd \ns \cd (\mue + \var \cd \ns) + \el \cd \nlv \cd (\mue + \var \cd (2 \nlv - \ns)))$$
Each individual term is positive: note that $\nlv > \ns$, so $2 \nlv - \ns > 0$. The overall term is negative, meaning that the derivative of the error is always negative. 
\end{proof}
Given these lemmas, we will divide the shape of the error curves into three components. Note that $err_{\si}(\s, \el = \el')$ and $err_{\li}(\s,\el = \el')$ are both curving downwards always. We will refer to the section of $err_{\si}(\s, \el = \el')$ with error greater than the error the small players get in $\pi(\Sv, 0)$ as $B$ and the region with error lower than this as $C$. Similarly, we will define the region of $err_{\li}(\s,\el = \el')$ where the error is higher than the error the large players get in $\pi(0, \Lv)$ as $B$, and the region where the error is lower as $C$. 

Note that, by Lemma \ref{spreflw},  $err_{\si}(\s = \s', \el)$ is always decreasing for all $\el > \el_0$: before that, it is increasing. We will call the region where the slope is increasing $A$, the portion where the slope is decreasing but still greater than the error the small player gets in $\pi(\Sv, 0)$ as $B$, and the final region, where the slope is decreasing and the error is lower than or equal to $\pi(\Sv, 0)$, as $C$.

We will not divide the regions of curve $err_{\li}(\s = \s', \el)$ because the proof below will not depend on those results. 

We will use as a shorthand $\sepcol = \{\pi(\Sv, 0), \pi(0, \Lv)\}$. 
\wslcore*
\begin{proof}
In this proof, we will use the results of Lemmas \ref{sprefsw}, \ref{spreflw}, \ref{lpreflw}, and \ref{lprefsw}. The statement of the theorem has three parts (including the equality $\gcol =_{\si} \pi(\Sv, 0)$), which we will prove in sequence. We note that Lemma \ref{lprefsw} indicates that $\pi(\Sv, \Lv) \cg_{\li} \pi(0, \Lv)$ always: the large players always prefer federating with the smaller players to being alone.  \\

\textbf{Case: $\pi(\Sv, \Lv) \cl_{\si} \pi(\Sv, 0)$}\\
In this case, we will show that $\sepcol = \{\pi(\Sv, 0), \pi(0, \Lv)\}$ is strictly core stable. 

We know from the precondition that the point $\pi(\Sv, \Lv)$ is in either region $A$ or $B$ of $err_{\si}(\s = \Sv, \el)$. We will show that any other arrangement $\pi(\s', \el')$ is one where the small players get higher error than in $\pi(\Sv, 0)$. We will assume $\el'>0$: if $\el'=0$ and $\s' < \Sv$, then we know $\pi(\s', 0) \cl_{\si} \pi(\Sv, 0)$ by Lemma \ref{sprefsw}. We will do this by starting at the point $\pi(\s', \el')$ and moving along the relevant error curves to towards $\pi(\Sv, 0)$ and $\pi(\Sv, \Lv)$ to show that the small players experience strictly greater error than in $\pi(\Sv, 0)$. 

First, we will consider the curve $err_{\si}(\s, \el = \el')$. From Lemma \ref{sprefsw}, we know that it is decreasing in $\s$. We will move $\s$ from $\s'$ to $\Sv$, which decreases error. (If $\s' = \Sv$, then this leaves error unchanged.) Next, we will consider the point $\pi(\Sv, \el')$ on the curve $err_{\si}(\s = \Sv, \el)$. If it is in region $A$ or $B$, then we know 
$$\pi(\s', \el') \cleq_{\si} \pi(\Sv, \el') \cl_{\si} \pi(\Sv, 0)$$
We will show that it is impossible for it to be in region $C$. If it were, then moving along this curve from $\el'$ to $\Lv$ would only decrease the error the small player experiences, because the slope is decreasing with $\el$ in region $C$. But then that would imply that $\pi(\Sv, \Lv)$ is in region $C$ of this graph, and so $\pi(\Sv, \Lv) \cg_{\si} \pi(\Sv, 0)$, which contradicts the precondition of this case.

This shows that the small players minimize their error in $\pi(\Sv, 0)$, so they will not even weakly prefer any other arrangement. Given this, any deviation from $\sepcol$ could only involve large players. However, by Lemma \ref{lpreflw}, out of any arrangement excluding small players, the large players minimize their error in $\pi(0, \Lv)$ (their current arrangement) so they would not even weakly prefer any other arrangement. 

This implies that $\sepcol$ is strictly core given $\pi(\Sv, \Lv) \cl_{\si} \pi(\Sv, 0)$. 

\textbf{Case: $\pi(\Sv, \Lv) =_{\si} \pi(\Sv, 0)$}\\
In this case, we assume that the small players are indifferent between $\gcol$ and $\sepcol$ and will show that $\gcol$ is strictly core stable. 

The key part of this proof is to show that, if the small players are indifferent between $\gcol$ and $\pi(\Sv, 0)$, then they get strictly greater error in every other arrangement $\pi(\s', \el')$ that isn't equal to $\pi(\Sv, 0)$ or $\pi(\Sv, \Lv)$.

The precondition means that $\pi(\Sv, \Lv)$ is in region $C$ of both graphs. To move to $\pi(\s', \el')$,  first we move along $err_{\si}(\s = \Sv, \el)$ from $\Lv$ to $\el'$. In doing so, we end up in either region $A$ or $B$ of the curve: note that we cannot remain in region $C$ unless $\el' = \Lv$. This is because by definition of the regions, region $C$ has error that decreases with $\el$, which means that error strictly increases from $\pi(\Sv, \Lv)$ when $\el$ is decreased. This means that if $\el' < \Lv$, $\pi(\Sv, \el')$ gives the small player greater error than the it gets in $\pi(\Sv, 0)$ (or equivalently, in $\gcol$). 

If the point $\pi(\Sv, \el')$ is in region $A$ or $B$ of $err_{\si}(\s = \Sv, \el)$, then it is also in region $B$ of the curve $err_{\si}(\s, \el = \el')$. Reducing $\s$ from $\Sv$ to $\s'$ can only increase the error the small player experiences. This shows that the small players get strictly greater error in any arrangement that is not $\pi(\Sv, 0)$ or $\pi(\Sv, \Lv)$.

Note: if $\el' = \Lv$, then point $\pi(\Sv, \el') = \pi(\Sv, \Lv)$ is in region $C$ of both curves, but because $\pi(\s', \el') \ne \pi(\Sv, \Lv)$, we must have $\s' < \Sv$, so $\pi(\s', \el') \cl_{\si} \pi(\Sv, \Lv) =_{\si} \pi(\Sv, 0)$. 

Next, we will show that $\gcol$ is strictly core stable.

As shown above, the small players get higher error in any arrangement besides $\gcol$ and $\pi(\Sv, 0)$. If they are in $\gcol$, they cannot be convinced to deviate to any arrangement except $\pi(\Sv, 0)$, and then only weakly (they get identical error). Conditional on the small players all being in $\gcol$, the large players would most prefer to be in $\gcol$, since by Lemma \ref{lpreflw} they strictly prefer it to $\pi(0, \Lv)$. This implies that there does not exist a group of players where all weakly would like to defect and at least one strictly wishes to defect, the condition for strict core stability. 

\textbf{Case: $\pi(\Sv, \Lv) \cg_{\si} \pi(\Sv, 0)$}: \\
We will show that $\gcol$ is strictly core stable. 

The precondition means that $\pi(\Sv, \Lv)$ is in section $C$ of both small player graphs. Consider an arbitrary other coalition $\pi(\s', \el')$: we will show that the small players get strictly greater error here than in $\pi(\Sv, \Lv)$.

First, we start at the arrangement $\pi(\Sv, \Lv)$. Consider moving along the curve $err_{\si}(\s = \Sv, \el)$: we hold $\s$ constant and reduce $\el$ from $\Lv$ to $\el'$. We could end up in region $C, B,$ or $A$ of this curve.  

First, assume we end up in section $C$. This means that, as we've reduced $\el$, our error has been monotonically increasing (unless $\el' = \Lv$, in which case it is unchanged). We also know that, in the curve $err_{\si}(\s, \el = \el')$, the point $\pi(\Sv, \el')$ is in region $C$ of this graph. Next, we similarly move along this curve to reduce $\s$ from $\Sv$ to $\s'$. From Lemma \ref{sprefsw}, we know $err_{\si}(\s, \el = \el')$ is monotonically decreasing in $\s$, which means reducing $\s$ monotonically increases the small player's error. (This is unless $\s' = \Sv$, in which case the error is unchanged.). Because we have assumed either $\s' \ne \Sv$ or $\el' \ne \Lv$, we have produced a path of monotonically increasing error from $\pi(\Sv, \Lv)$ to $\pi(\s', \el')$, so we know the small player experiences greater error here. 

Next, we will instead assume that we ended up in region $A$ or $B$ of the $err_{\si}(\s = \Sv, \el)$ curve when we reduced $\el$. By definition of the $A$ and $B$ regions, we know that the small player experiences larger error in $\pi(\Sv, \el')$ than it would in $\pi(\Sv, 0)$. (The one exception is the point $\pi(\Sv, 0)$, which is in region $A$ and which obviously gives equal error to $\pi(\Sv, 0)$).  Then, we know that in the curve $err_{\si}(\s, \el = \el')$, the point $\pi(\Sv, \el')$ must be in region $B$. Because $err_{\si}(\s, \el = \el')$ is monotonically decreasing in $\s$, reducing $\s$ from $\Sv$ to $\s'$ monotonically increases error. We then know that: 
$$\pi(\s', \el') \cleq_{\si} \pi(\Sv, \el') \cleq_{\si} \pi(\Sv, 0) \cl_{\si} \pi(\Sv, \Lv)$$
where the last inequality comes from the premise of this statement. The first and second inequalities are strict if $\s' < \Sv$ or $\l' \ne 0$, respectively. 

We have just shown that the small player prefers $\pi(\Sv, \Lv)$ to any other arrangement. If there were going to be a defecting group, it would have to involve only large players. However, the arrangement where large players (federating by themselves) get lowest error is in $\pi(0, \Lv)$, which by Lemma \ref{lprefsw} gives them higher error than $\pi(\Sv, \Lv)$.

By this reasoning, $\gcol$ is strictly core stable. 
\end{proof}

\section{Supporting proofs for fine-grained federation}\label{app:vsupp}

\vfedminval*
\begin{proof}
To minimize, we will take the derivative of player $j$'s error with respect to the $\vm_{jk}$ weight. Note that we only have $\vm_{jj} = 1 - \sum_{i\ne j}\vm_{ji} = 1 - \vm_{jk} - \sum_{i\ne j, i\ne k}\vm_{ji}$ so $\vm_{jk}$ appears twice in the component involving $\mue$. Rewriting the error gives: 
$$\mue \sum_{i\ne j}^{\nplayer}\frac{\vm_{ji}^2}{\ndraw_i} + \mue \frac{(1-\sum_{i\ne j}\vm_{ji})^2}{\ndraw_j}$$
$$+ \p{\sum_{i\ne j}\vm_{ji}^2 + \p{\sum_{i\ne j}\vm_{ji}}^2} \cd \var$$
Taking the derivative with respect to $\vm_{jk}$ gives: 
$$\mue \frac{2\vm_{jk}}{\ndraw_k} - 2\mue \frac{1-\sum_{i\ne j}\vm_{ji}}{\ndraw_j} + \var \p{2\sum_{i\ne j} \vm_{ji} + 2\vm_{jk}}$$
To confirm that we are finding a minimum rather than a maximum, we note that the second derivative is always positive: 
$$\mue \frac{2}{\ndraw_k} + \mue \frac{2}{\ndraw_j} + \var \p{2+ 2}>0$$
We first simplify the derivative by substituting in for $\vm_{jj}$: 
$$\mue \frac{2\vm_{jk}}{\ndraw_k} - 2\mue \frac{\vm_{jj}}{\ndraw_j} + 2\var \p{1-\vm_{jj} + \vm_{jk}}=0$$
And then solve for $\vm_{jk}$ to obtain: 
$$\vm_{jk} = \frac{\vm_{jj}\cd \p{\var + \frac{\mue}{\ndraw_j}} - \var}{\var + \frac{\mue}{\ndraw_k}}$$
To find $\vm_{jj}$, we use that all of the weights sum up to 1: 
$$\vm_{jj} + \sum_{i\ne j}\frac{\vm_{jj}\cd \p{\var + \frac{\mue}{\ndraw_j}} - \var}{\var + \frac{\mue}{\ndraw_k}} = 1$$
$$\vm_{jj} + \sum_{i\ne j}\frac{\vm_{jj}\cd \p{\var + \frac{\mue}{\ndraw_j}}}{\var + \frac{\mue}{\ndraw_k}} -\sum_{i\ne j}\frac{\var}{\var + \frac{\mue}{\ndraw_k}} = 1$$
$$\vm_{jj}\p{1 + \sum_{i\ne j}\frac{\var + \frac{\mue}{\ndraw_j}}{\var + \frac{\mue}{\ndraw_k}}} -\sum_{i\ne j}\frac{\var}{\var + \frac{\mue}{\ndraw_k}} = 1$$
$$\vm_{jj} = \frac{1 + \sum_{i\ne j}\frac{\var}{\var + \frac{\mue}{\ndraw_k}}}{1 + \sum_{i\ne j}\frac{\var + \frac{\mue}{\ndraw_j}}{\var + \frac{\mue}{\ndraw_k}}}$$
Next, we define $V_i = \var + \frac{\mue}{\ndraw_i}$. This allows us to rewrite the term as: 
$$\vm_{jj} = \frac{1 + \var \sum_{i\ne j}\frac{1}{V_i}}{1 + V_j\sum_{i\ne j}\frac{1}{V_i}}$$
Similarly, we can rewrite: 
$$\vm_{jk} = \frac{1}{V_k}\cd \frac{V_j-\var}{1 + V_j \sum_{i\ne j}\frac{1}{V_i}}$$
\end{proof}

\end{document}